\documentclass[12pt, reqno]{amsart}

\title[Optical fiber modes by FEAST]
{Computing leaky modes of optical fibers using a FEAST algorithm
  for polynomial eigenproblems}

\author{J.~Gopalakrishnan}
\address{Portland State University, PO Box 751, Portland OR 97207,USA }
\email{gjay@pdx.edu}

\author{B.~Q.~Parker}
\address{Portland State University, PO Box 751, Portland OR 97207,USA }
\email{bqp2@pdx.edu}

\author{P.~VandenBerge}
\address{Portland State University, PO Box 751, Portland OR 97207,USA }
\email{piet2@pdx.edu}

\usepackage{amsmath,amssymb}
\usepackage{bm}
\usepackage{tikz}
\usepackage{pgfplotstable, pgfplots}
\usepackage{subcaption}
\usepackage{mathabx}
\usetikzlibrary{shapes,arrows,snakes,calendar,matrix,backgrounds,folding,calc,positioning,spy,decorations.pathreplacing}
\usepackage{hyperref}
\usepackage{siunitx}
\usepackage{url}
\usepackage{pseudo}
\usepackage{float}
\floatstyle{ruled}
\newfloat{alg}{thp}{lop}
\floatname{alg}{Algorithm}

\theoremstyle{plain}
\newtheorem{theorem}{Theorem}
\newtheorem{lemma}[theorem]{Lemma}

\theoremstyle{remark}
\newtheorem{remark}[theorem]{Remark}
\newtheorem{problem}[theorem]{Problem}

\newcommand{\C}{\mathbb{C}}
\newcommand{\R}{\mathbb{R}}

\newcommand{\Ac}{\mathcal{A}}
\newcommand{\Bc}{\mathcal{B}}
\newcommand{\Kc}{\mathcal{K}}
\newcommand{\Ec}{\mathcal{E}}
\newcommand{\Ect}{\tilde{\mathcal{E}}}
\newcommand{\Sc}{\mathcal{S}}
\newcommand{\Sct}{\tilde{\mathcal{S}}}

\newcommand{\Vc}{\mathcal{V}}
\newcommand{\Vct}{\tilde{\mathcal{V}}}

\newcommand{\ii}{\hat{\imath}}

\newcommand{\ut}{\tilde{u}}
\newcommand{\uh}{\hat{u}}
\newcommand{\vt}{\tilde{v}}
\newcommand{\xt}{\tilde{x}}
\newcommand{\xh}{\hat{x}}

\newcommand{\Rt}{\tilde{R}}
\newcommand{\St}{\tilde{S}}
\newcommand{\Sm}[1]{S^{({#1})}}
\newcommand{\Pt}{\tilde{P}}
\newcommand{\Et}{\tilde{E}}
\newcommand{\Xt}{\tilde{X}}
\newcommand{\Yt}{\tilde{Y}}
\newcommand{\Vt}{\tilde{V}}
\newcommand{\Wt}{\tilde{W}}
\newcommand{\vL}{\varLambda}
\newcommand{\veps}{\varepsilon}
\newcommand{\diag}{\mathop{\mathrm{diag}}}
\newcommand{\diam}{\mathop{\mathrm{diam}}}
\renewcommand{\Gamma}{\varGamma}
\renewcommand{\Lambda}{\varLambda}
\renewcommand{\vL}{\varLambda}
\newcommand{\vG}{\varGamma}
\newcommand{\ran}{\mathop{\mathrm{ran}}}
\newcommand{\Rh}{{\hat{R}}}
\newcommand{\rh}{{\hat{r}}}
\newcommand{\rt}{{\tilde{r}}}
\newcommand{\ec}{e_{\mathrm{cap}}}
\newcommand{\tc}{t_{\mathrm{cap}}}
\newcommand{\dc}{d_{\mathrm{cap}}}
\newcommand{\Rci}{R_{i, \mathrm{cap}}}
\newcommand{\Rco}{R_{o, \mathrm{cap}}}
\newcommand{\Rcore}{R_{\mathrm{core}}}
\newcommand{\tclad}{t_{\mathrm{clad}}}
\newcommand{\Rout}{\Rh_{\mathrm{fin}}}
\newcommand{\Rfin}{\Rh_{\mathrm{fin}}}
\newcommand{\om}{\varOmega}
\newcommand{\oh}{\varOmega_h}
\newcommand{\nair}{{n_{\mathrm{air}}}}
\newcommand{\nsi}{{n_{\mathrm{Si}}}}
\newcommand{\omsi}{{\varOmega_{\mathrm{Si}}}}
\newcommand{\ompml}{{\varOmega_{\mathrm{pml}}}}
\newcommand{\omint}{{\varOmega_{\mathrm{int}}}}
\newcommand{\Ho}{\mathring{H}}
\def\d{\partial}
\renewcommand{\Im}{\mathop{\mathrm{imag}}}
\renewcommand{\Re}{\mathop{\mathrm{real}}}
\newcommand{\jmp}[1]{\ldbrack{{#1}}\rdbrack}
\pgfplotsset{width=7cm,compat=1.12}
\usetikzlibrary{pgfplots.groupplots}

\allowdisplaybreaks   

\begin{document}


\begin{abstract}
  An efficient contour integral
technique to approximate a cluster of nonlinear eigenvalues
of a polynomial eigenproblem, circumventing
certain large inversions from a linearization, is presented.
It is applied to 
the nonlinear eigenproblem that arises from a frequency-dependent
perfectly matched layer.
This approach is shown to
result in an accurate method for computing leaky modes of optical
fibers.  Extensive computations on an antiresonant fiber with a
complex transverse microstructure are reported. This structure
is found to present substantial computational difficulties:
Even when employing
over one million degrees of freedom, the fiber model appears to remain in
a preasymptotic regime where computed confinement loss values are
likely to be off by orders of magnitude. Other difficulties
in computing mode losses, together with practical techniques to overcome
them, are detailed.


\end{abstract}


\maketitle

\section{Introduction}   \label{sec:introduction}

In this paper, we bring together recent advances in contour integral
eigensolvers and perfectly matched layers to improve techniques for
computing transverse modes of optical fibers. Unlike classical
step-index optical fibers, many emerging microstructured optical
fibers do not have perfectly guided modes. Yet, they can quite
effectively guide energy in leaky modes,
also known as quasi-normal modes or
resonances.  Confinement losses of leaky modes, and their accurate
computation, are of considerable practical importance. We approach
this computation by solving a nonlinear (polynomial) eigenproblem
obtained using a frequency-dependent perfectly matched layer (PML) and
high order finite element discretizations.

The PML we use is the one recently studied
by~\cite{NannenWess18}. Although their essential idea is the same as
the early works on PML~\cite{Beren94, ChewWeedo94,ColliMonk98}, their
work is better appreciated in the following context.  While adapting
the PML for source problems to eigenproblems,
many~\cite{AraujEngst17, GopalMoskoSanto08, KimPasci09} preferred a
frequency-independent PML over a frequency-dependent PML. This is
because for eigenproblems obtained using PML, the ``frequency'' is
related to the unknown eigenvalue, so a frequency-dependent approach
results in equations with a nonlinear dependence on the unknown
eigenvalue, and hence a nonlinear eigenproblem.  In contrast, a
frequency-independent approach results in a standard linear
generalized eigenproblem, for which many standard solvers exist.
However, the authors of~\cite{NannenWess18} made a compelling case for
the use of frequency-dependent PML by showing an overall reduction in
spurious modes and improved preasymptotic eigenvalue approximations.
The price to pay in their approach is that instead of
a linear eigenproblem, one must solve a nonlinear (rational)
eigenproblem. Excluding a zero singularity, one can reduce this to a
polynomial eigenproblem.

One of the goals of this paper is to show that such polynomial
eigenproblems can be solved using a contour integral
eigensolver, recently popularized in numerical linear algebra under
the name ``the FEAST algorithm'' \cite{GuttePolizTang15,
  Poliz09}. Accordingly, we begin our study in Section~\ref{sec:feast}
by introducing the algorithm and our adaptation of it to
polynomial eigenproblems. We use a well-known
linearization~\cite{GohLanRod82} of a degree $d$ polynomial
eigenproblem to get a linear eigenproblem with $d$ times as many
unknowns as the original nonlinear eigenproblem.  This $d$-fold
increase in size is prohibitive, especially for applications like the
computation of leaky optical modes, which, as we shall see in
Section~\ref{sec:microfiber}, will need several millions of degrees of
freedom (before linearization).  In Section~\ref{sec:feast}, we show
how to overcome this problem. Exploiting the fact that FEAST only
requires the application of the resolvent of the linearization,
we develop an identity for this large linearized resolvent in terms of
smaller nonlinear resolvents of the original size. This leads to our
eigensolver presented in Algorithm~\ref{alg:polyfeast}.

In Section~\ref{sec:leaky}, we formulate the equations for the
transverse leaky modes of general optical fibers nondimensionally,
introduce the equations of frequency-dependent PML, use an arbitrary
order finite element discretization, and present the resulting cubic
eigenproblem. The algorithm developed in the previous section is then
applied. An interesting feature that derives from the combination of this
discretization with our algorithm
is that any spurious mode formed of finite element functions
supported only in the PML region is automatically eliminated from the
output eigenspace. This is because the formulation sends
such functions to the eigenspace of~$\infty$. This section also has a
verification of the correctness of our approach using a semianalytical
calculation for leaky modes of step-index fibers.

\begin{figure}
  \centering
  \begin{subfigure}{0.5\linewidth}
    \begin{tikzpicture}[scale=0.06]


      \fill [blue!30,even odd rule,opacity=0.6] (0,0)
      circle[radius=50]
      circle[radius=37.5];


      \draw[red, dotted] (14,0)
      arc [radius=14.0, start angle=0, end angle=30];
      \draw[red, dotted] (14,0)
      arc [radius=14.0, start angle=0, end angle=-30];

      \draw[red, ->] (0,0) -- (14,0)
      node [scale=0.75, black, midway, below]
      {{{$\Rcore$}}};

      \draw[red,<->] (37.5, 0)--++(12.5, 0)
      node [scale=0.75, black, midway, below]
      {{{$\tclad$}}};

      \draw[red, ->] (0,0) -- ++(60:56) 
      node [scale=0.75,  pos=0.615, above, sloped, black]
      {{$R_0$}};

      \draw[red, dotted] (60:56) 
      arc [radius=56.0, start angle=60, end angle=70];
      \draw[red, dotted] (60:56) 
      arc [radius=56.0, start angle=60, end angle=0];

      \fill [blue!30,even odd rule,opacity=0.6] (0.0, 26.505) 
      circle [radius=12]
      circle [radius=10];
      

      
      
      
      \fill [blue!30,even odd rule,opacity=0.6] (-22.954003327306545, 13.2525) 
      circle [radius=12]
      circle [radius=10];
      
      \fill [blue!30,even odd rule,opacity=0.6] (-22.954003327306545,
      -13.2525) circle [radius=12] circle [radius=10];
      
      \draw[<->, red]
      (-22.954003327306545, 1.2525) -- (-22.954003327306545, -1.2525)
      node [scale=0.7, black, midway, right] {$\dc$};
      
      \fill [blue!30,even odd rule,opacity=0.6] (0.0, -26.505) 
      circle [radius=12]
      circle [radius=10];

      \fill [blue!30,even odd rule,opacity=0.6] (22.954003327306545, -13.2525) 
      circle [radius=12]
      circle [radius=10];

      \fill [blue!30,even odd rule,opacity=0.6] (22.954003327306545,
      13.2525) circle [radius=12] circle [radius=10];
    \end{tikzpicture}
    \subcaption{Shaded and white areas indicate glass and
      air, respectively}
    \label{fig:geom}
  \end{subfigure}
  \qquad 
  \begin{subfigure}{0.4\linewidth}
    \begin{tikzpicture}
      [scale=0.15, show background rectangle,inner frame sep=0mm]

      \begin{scope}
        
        \clip (-20, 20) rectangle (20, 45); 
        

        \fill [blue!30,even odd rule,opacity=0.6] (0,0)
        circle[radius=60]
        circle[radius=37.5];


        \draw[red, dotted] (0,0) circle [radius=14.0];
        \draw[red, ->, dotted] (0,0) -- (14,0)
        node [scale=0.75,black, midway, above,fill=white] {{\tt{Rc}}};
        \fill[red] (0,0) circle (20pt);
        
        \fill [blue!30,even odd rule,opacity=0.6] (0.0, 26.505) 
        circle [radius=12]
        circle [radius=10];
        

        
        
        \draw[red, <->] (7.071067811865475, 33.57606781186547) --
        (8.48528137423857, 34.99028137423857) node [scale=0.75, black,
        above, right] {$\tc$}; 
        
        \fill [blue!30,even odd rule,opacity=0.6] (-22.954003327306545, 13.2525) 
        circle [radius=12]
        circle [radius=10];
        
        \fill [blue!30,even odd rule,opacity=0.6] (-22.954003327306545,
        -13.2525) circle [radius=12] circle [radius=10];
        
        \draw[<->, red]
        (-22.954003327306545, 1.2525) -- (-22.954003327306545, -1.2525)
        node [scale=0.7, black, midway, right] {$\dc$};
        
        \fill [blue!30,even odd rule,opacity=0.6] (0.0, -26.505) 
        circle [radius=12]
        circle [radius=10];

        \fill [blue!30,even odd rule,opacity=0.6] (22.954003327306545, -13.2525) 
        circle [radius=12]
        circle [radius=10];

        \fill [blue!30,even odd rule,opacity=0.6] (22.954003327306545,
        13.2525) circle [radius=12] circle [radius=10];



        \draw[->,red] (0.0, 35.505) -- (0.0, 37.505);
        \draw[->,red] (0.0, 40.505) -- (0.0, 38.505)
        node [scale=0.75, black, right] {$\ec$};

        \draw[->, red] (0.0, 26.505) -- (10.0, 26.505)
        node [scale=0.75, black, midway, above] {$\Rci$};

        \draw[red, ->] (0.0, 26.505) --++(-20:12)
        node [scale=.75, black, midway, sloped, below] {$\Rco$}; 
      \end{scope}
    \end{tikzpicture}
    \subcaption{Zoomed in view near the top capillary tube}
    \label{fig:geom-zoom}
  \end{subfigure}
  \caption{Transverse geometry of a microstructured 
    fiber~\cite{KolyaKosolPryam13,Polet14,YuKnigh16}}
  \label{fig:microgeom}
\end{figure}
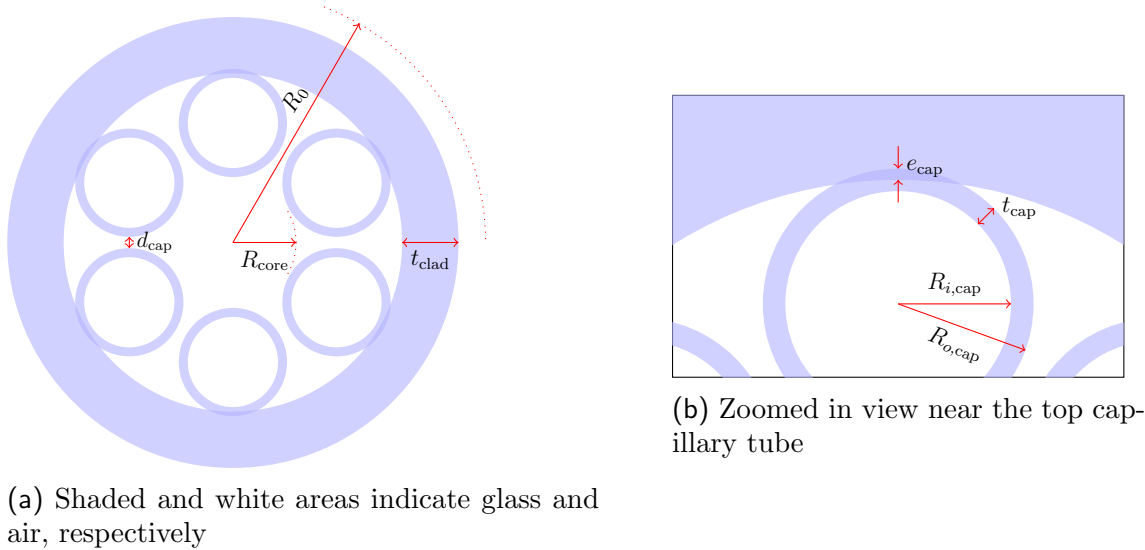

In Section~\ref{sec:microfiber}, we consider a
microstructured optical fiber. Recent microstructured fibers fall into
two categories: photonic band gap fibers, and antiresonant
fibers. Emerging fibers of the latter class seem not to have received
much attention in the mathematical literature, although they are
actively pursued in the optics literature~\cite{KolyaKosolPryam13,
  Polet14, YuKnigh16}. Antiresonant optical fiber designs with
air-filled hollow cores are particularly interesting since dispersion,
nonlinear optical effects, and propagation losses are all negligible
in air.  We study such a fiber in Section~\ref{sec:microfiber},
providing enough detail in the hope that it may serve as a benchmark
problem for others. The fiber geometry is illustrated in
Figure~\ref{fig:microgeom}.

Some of the difficulties we encountered while computing the mode
losses in Section~\ref{sec:microfiber} are worth noting here.
For microstructured fibers with thin structural elements, we have
generally found it difficult to find perfect
agreement between our converged
loss values and those
reported in the optics literature produced using proprietary
software.
Our results in Section~\ref{sec:microfiber}
illuminate the issue.
For the fiber we considered, we found a surprisingly large
preasymptotic regime where confinement loss values jump orders of
magnitude when mesh size ($h$) and finite element degree ($p$) are
varied. Hence it seems possible to find agreement with whatever
loss value in the literature, experimental or
numerical, by simply adjusting model and PML parameters, while
the discretization is in the preasymptotic regime. But such agreement
is meaningless. 
In view of the results of Section~\ref{sec:microfiber},
we cannot recommend trusting
computed
resonance values and confinement losses reported without any evidence
of them having stabilized over variations in $h$ and $p$ (or unsupported by
other convergence studies).  We will show multiple routes to get to
the asymptotic regime of converging eigenvalues by either decreasing
$h$ or by increasing $p$. In our experience, quicker routes to this
asymptotic regime are generally offered by the latter.

While searching for core modes in such microstructured fibers, one
should be wary of modes that carry energy in structures outside of the
hollow core. Although these are unwanted modes, they are not spurious
modes---they are actual eigenmodes of the
structure. (Figure~\ref{fig:highlossmodes} shows such unwanted
non-core modes; cf.~core modes in Figure~\ref{fig:modes_poletti}.)
Another issue, well-recognized by
many~\cite{AraujEngst17,GopalMoskoSanto08,KimPasci09,NannenWess18} in
other resonance computations, is
the interference of spurious modes that arise from the discretization
of the essential spectrum (and Figure~\ref{fig:search} in this paper
also provides a glimpse of
this issue).  In Section~\ref{sec:microfiber}, we
indicate a way to overcome this problem to some extent by using an
elliptical contour in our eigensolver.  Considering the
expected deformation of the essential spectrum due to the PML, one may
adjust the eccentricity of the ellipse to probe spectral regions of
interest fairly close to the origin without wasting computational
resources on unwanted eigenfunctions from the deformed essential
spectrum.

In Section~\ref{sec:proof}, we present 
proofs of the two theorems presented in the next section. We close
with concluding remarks in Section~\ref{sec:conclusion}.

\section{A FEAST algorithm for polynomial eigenproblems}
\label{sec:feast}

Consider the problem of finding a targeted cluster of nonlinear
eigenvalues, enclosed within a given contour, and its associated
eigenspace. (The problem of interest  is precisely stated
as Problem~\ref{prb:nonlin} in
Subsection~\ref{ssec:polyneigen} below.)
The FEAST algorithm \cite{GuttePolizTang15,
  KestyPolizTang16, Poliz09} is one type of contour integral
eigensolver that addresses such problems.
A version of the algorithm for nonlinear eigenproblems
was presented in~\cite{GavMiePol18}, but here we shall pursue specific
simplifications possible when the nonlinearity is of polynomial type.
We begin by describing the standard FEAST algorithm for linear
eigenproblems.  Consider $A, B \in \C^{n \times n}$ and the linear
generalized eigenproblem of finding numbers $\lambda \in \C,$ an
associated (right) eigenvector $0 \ne x \in \C^n$, and a left
eigenvector $0 \ne \xt \in \C^n$, satisfying
\begin{equation}
  \label{eq:4}
  \xt^* A = \lambda \xt^* B, \qquad A x = \lambda B x.  
\end{equation}
We shall also consider left and right
generalized eigenvectors of such eigenproblems, since they are needed
to formulate an accurate relation between ranges of certain
spectral projectors at the foundation of the algorithm. Recall that 
the left and right algebraic eigenspaces (or generalized eigenspaces)
of a linear eigenvalue
are, respectively, the spans of its left and right generalized
eigenvectors. Generalizations of these concepts to the nonlinear
case are defined in Subsection~\ref{ssec:polyneigen}.

\subsection{Spectral projector approximation}
\label{sec:stdfeast}
 
Suppose we want to compute a cluster of eigenvalues, collected into a
set $\vL$, and its accompanying (right) algebraic eigenspace, denoted
by $E \subset \C^n$, and left algebraic eigenspace $\Et \subset \C^n$.
The wanted eigenvalues, namely elements of $\Lambda,$ are known to be
enclosed within $\Gamma$, a positively oriented, bounded, simple, closed
contour that does not cross any eigenvalue.

The matrix-valued integrals
\begin{align}
  \label{eq:SpectralProjection}
  S=\frac{1}{2\pi\ii}\oint_\Gamma (z B - A)^{-1} B \, dz, \qquad
  \St =\frac{1}{2\pi\ii}\oint_{\Gamma} (z B - A)^{-*} B^* \, dz,
\end{align}
sometimes called Riesz projections, or spectral projectors, are well
known to yield projections onto the eigenvalue cluster's right and
left algebraic eigenspaces $E$ and $\Et$, respectively (see
e.g.,~\cite[page~39]{Kato95} or \cite[Theorem~1.5.4]{Davie07}).
Here and throughout, we use $^*$ and $'$ to denote conjugate transpose
and transpose, respectively, so $M^* = \bar{M}'$ for any matrix $M$
and we abbreviate $(M^*)^{-1}$ to $M^{-*}$ for invertible $M$.
Focusing on $E$ for the moment, $E$ is the eigenspace of $S$
associated to its eigenvalue one. The only other
eigenvalue of $S$ is zero. The
algebraic eigenspaces of all the eigenvalues of~\eqref{eq:4} not
enclosed by $\vG$ have been mapped to the eigenspace of the zero
eigenvalue of~$S$.  Hence, if we can compute $S$, then a well-known
generalization of the power iteration (namely the subspace iteration)
when applied to $S$ will converge to $E$ in one iteration.  Along the
same lines, we also conclude that a subspace iteration with $\St$ will
converge at once to~$\Et$.

The FEAST algorithm is simply a subspace iteration,
performed after replacing $S$ and $\St$ by computable quadrature
approximations.
These quadrature approximations of $S$ and $\St$ take the form
\begin{equation}
  \label{eq:SN}
  S_N = \sum_{k=0}^{N-1} w_k (z_k B - A)^{-1}  B, 
  \qquad \St_N = \sum_{k=0}^{N-1} \bar{w}_k (z_k B - A)^{-*}  B^*,  
\end{equation}
for some $w_k \in \C$ and points $z_k \in \Gamma$. Then, the
mathematical statement of the FEAST algorithm is as follows: given initial
right and left subspaces $E_0, \Et_0 \subset \C^n$, compute two sequences
of subspaces, $E_\ell$ and $\Et_\ell$, by
\begin{equation}
  \label{eq:5}
  E_\ell = S_N E_{\ell-1},
  \qquad
  \Et_\ell = \St_N \Et_{\ell-1}
  \qquad \text{ for } \ell = 1, 2, \ldots.  
\end{equation}
Here and throughout
we use
$M Z$ to denote $\{ M z: z \in Z\}$
for a matrix $M \in \C^{n \times n}$ and a
{\em subspace} $Z \subseteq \C^n$. A practical implementation of the FEAST
algorithm is not as simple as~\eqref{eq:5} because it must
additionally take care of normalization and computation of Ritz values
for each subspace iterate~\cite{GuttePolizTang15, KestyPolizTang16,
  Poliz09}.

In this paper, we will use only circular and elliptical
contours for $\Gamma$. Both have been 
studied previously~\cite{GopalGrubiOvall20a, GuttePolizTang15}, so we
will be brief.
Letting
$\phi = \pi/N$, we 
parametrize a circle $\Gamma$ of radius
$\gamma>0$ centered at $y \in \C$, in terms of $\theta$, by
$\Gamma = \{ \gamma \exp( \ii (\theta + \phi) + y : 0 \le \theta < 2
\pi\}$. Transforming the integrals over $z$
in~\eqref{eq:SpectralProjection} as integrals over $\theta$, and then
applying the trapezoidal rule at equally spaced
$N$ values of $\theta$, shifted by $\phi$, we obtain
a quadrature approximation of $S$ as in~\eqref{eq:SN}, with 
\begin{equation}
  \label{eq:zkwk}
  \theta_k = \frac{2 \pi k}{N}, \quad
  z_k = \gamma \exp(\ii (\theta_k + \phi)) + y,
  \quad
  w_k = \frac{\gamma}{N} \exp(\ii(\theta_k + \phi).  
\end{equation}
By estimating the separation between
wanted and unwanted eigenvalues after the spectral mapping by $S_N$,
it is possible to estimate the rate of convergence of the subspace
iteration~\eqref{eq:5} while employing this quadrature (see e.g.,
\cite[Example~2.2]{GopalGrubiOvall20a} or~\cite{GuttePolizTang15}).

We will also use  elliptical contours later.
Letting $\gamma>0, \rho>1, y \in \C$,  we restrict
ourselves to Bernstein ellipses
$\Gamma = \{y + \gamma e^{\ii \theta}\rho /(\rho+\rho^{-1})
+ \gamma e^{-\ii \theta}
\rho^{-1}/(\rho+\rho^{-1}): 0\le \theta < 2 \pi\}$ 
aligned with the coordinate axes.
We again use
an $N$~point uniform trapezoidal rule, shifting the ellipse
parametrization $\theta$ by $\phi =  \pi/N$. Simple computations then
lead to the formulas
\begin{equation}
  \label{eq:zkwk-ellipse}
  z_k = y +
  \gamma \frac{\rho e^{\ii (\theta_k+\phi) } +
    \rho^{-1}e^{-\ii(\theta_k+\phi)} } {\rho+\rho^{-1}},
  \quad
  w_k = \gamma
  \frac{\rho e^{\ii (\theta_k+\phi) }
    - \rho^{-1} e^{-\ii(\theta_k+\phi)} }{N(\rho+\rho^{-1})}
\end{equation}
where $\theta_k = 2\pi k /N$. These are the values we shall use
in~\eqref{eq:SN} when we need elliptical contours later.

\subsection{Polynomial eigenproblems}
\label{ssec:polyneigen}

In this subsection we establish notation for
polynomial eigenproblems and consider a nonlinear eigenvalue cluster
approximation problem.
Suppose we are given $d+1$ matrices
$A_i \in \C^{n \times n}$, $i=0, \ldots, d$.
We assume that the last matrix 
$A_d$ is nonzero (in order to fix the grade $d$), but do
not assume that $A_d$ is invertible.
Let $\C^+ = \C \cup \{\infty\}$, the extended complex plane.
We consider the {\em
  polynomial eigenproblem} of finding a
$\lambda \in \C^+$ satisfying
\begin{subequations}
  \label{eq:polyeigprob}
  \begin{equation}
    \label{eq:1}
    P(\lambda) x = 0, \quad \xt^* P(\lambda) = 0 
  \end{equation}
  for some nontrivial $x, \xt \in \C^n$, called the nonlinear right and
  left 
  eigenvectors, respectively. Here 
   $P(z)$ is a matrix polynomial of degree $d$,  given by
  \begin{equation}
    \label{eq:2}
    P(z) = \sum_{j=0}^d  z^j A_j. 
  \end{equation}
\end{subequations}
Since $P(z)$ is  a square matrix, $\lambda$ may equivalently
be thought of as a root of the nonlinear
equation $\det P(z) =0$. The algebraic multiplicity of the
nonlinear eigenvalue $\lambda$ is its multiplicity as a root of
the polynomial $\det P(z)$, a polynomial which we
assume does not  vanish everywhere.  Note that
$\lambda = \infty$ is said to be
an eigenvalue of~\eqref{eq:polyeigprob}
if zero is an eigenvalue of
$z^dP(z^{-1}) = z^d A_0 + z^{d-1} A_1 + \cdots + A_d$ \cite{GutteTisse17, TisMee01}.

Next, let $P^{(l)}(z)$ denote the $l^{\text{th}}$ derivative
($d^l P/ dz^l$) of $P$ with respect to the complex variable $z$.
Ordered sequences
$x_0, x_1, \ldots, x_{k-1}$ and
$\xt_0, \xt_1, \ldots, \xt_{k-1}$
in $\C^n$
are respectively
called~\cite{GutteTisse17} right and left Jordan chains for the matrix
polynomial $P$ at $\lambda$ if
\begin{equation}
  \label{eq:18}
  \sum_{l=0}^j \frac 1 {l!} P^{(l)}(\lambda)  x_{j-l} = 0,
  \qquad
  \sum_{l=0}^j \frac 1 {l!} \xt_{j-l}^* P^{(l)}(\lambda)  = 0,
  \qquad j=0, 1, \ldots, k-1.  
\end{equation}
When $k=1$, the chains reduce to singletons and~\eqref{eq:18}
coincides with the equation for an eigenvector~\eqref{eq:1}.  For more
general $k$, the vectors of these chains are referred to as (right and
left) nonlinear generalized eigenvectors.  The right and left {\em
  algebraic eigenspaces of a set of nonlinear eigenvalues} $\vL$ are,
respectively, the span of all the right and left nonlinear generalized
eigenvectors associated to every~$\lambda$ in $\vL$.  These
definitions generalize the standard notion of algebraic eigenspace for
the linear eigenproblem: indeed, in basic linear algebra, one
respectively
calls
the sequences $x_0, x_1, \ldots, x_{k-1} \in \C^n$
and $\xt_0, \xt_1, \ldots, \xt_{k-1} \in \C^n$
a right and left Jordan chain of
$A - \lambda B \in \C^{n \times n}$ if, for all  $i=1, 2,\ldots, k-1,$
\begin{subequations}
  \begin{gather}
      \label{eq:17}
      (A - \lambda B) x_0 = 0,
       \text{ and }
      (A - \lambda B) x_i = Bx_{i-1},
      \\\label{eq:17adj}
      \xt_0^*(A - \lambda B) =0,
      \text{ and }
      \xt_i^*( A- \lambda B) = \xt_{i-1}^* B.
  \end{gather}
\end{subequations}
It is easy to see that~\eqref{eq:17} and~\eqref{eq:17adj}
are respectively equivalent to the first and second 
equalities of~\eqref{eq:18}, when $P(\lambda)$ is set to the linear
matrix polynomial $A - \lambda B$.  With these notions, we can state
the nonlinear analogue of the eigenvalue cluster approximation problem
considered in Section~\ref{sec:stdfeast}.

\begin{problem}
  \label{prb:nonlin}
  Compute a cluster $\vL$ of nonlinear eigenvalues of $P(z)$ enclosed
  within  $\Gamma$ and its accompanying right and left algebraic
  eigenspaces $E$ and $\Et$, respectively.
\end{problem}

In the study of matrix polynomials, the concept of a linearization is
crucial~\cite{GohLanRod82}. The first companion linearization of
$P(z)$ is the matrix pencil $\Ac - z \Bc \in \C^{nd \times nd}$, shown
below in a $d \times d$ block partitioning where the blocks are
elements of $ \C^{n \times n}$:
\begin{equation}
  \label{eq:15}
  \Ac
  =
  \begin{bmatrix}
    0    &  I   & 0       & \cdots & 0 \\
    0    &  0   & I     &  \ddots      &  \vdots \\
 \vdots  &\vdots & \ddots  & \ddots & 0\\
   0     & 0    &  \cdots &     0  & I \\
   A_0   & A_1  & \cdots  & A_{d-2} & A_{d-1}
 \end{bmatrix},
 \quad
 \Bc
   = 
  \begin{bmatrix}   
    I    & 0      & \cdots  & \cdots &   0 \\
    0    & I      & \ddots  &        & \vdots \\
    \vdots& \ddots& \ddots  & \ddots     & \vdots\\
    \vdots&       & \ddots  & I & 0 \\
    0 &   \cdots  & \cdots  & 0 & -A_d
 \end{bmatrix}.  
\end{equation}
Here and throughout,
$I$ denotes the identity matrix (whose dimensions may differ
 at different occurrences, but will always be clear from context).
Let us note a well known connection between the nonlinear
eigenproblem~\eqref{eq:polyeigprob} and the linear eigenproblem
\begin{equation}
  \label{eq:6}
  \Ac X = \lambda \Bc X, \qquad
  \Xt^*\Ac = \lambda \Xt^* \Bc,
\end{equation}
for nontrivial  left and right eigenvectors $X$ and $\Xt$,
respectively,
and a ``linear'' eigenvalue~$\lambda$.
We block partition 
$Y  \in \C^{nd \times m}$ using blocks $Y_i$ in $ \C^{n \times m}$
(where the $m=1$ case represents  a block partitioning of column vectors)
as shown below, where we also define $F \in \C^{n \times nd}$ and
$L \in
\C^{n \times nd},$
all using a block partitioning  compatible with~\eqref{eq:15}:
\begin{equation}
  \label{eq:19}
  Y =
  \begin{bmatrix}
    Y_0 \\ Y_1 \\ \vdots \\ Y_{d-1}
  \end{bmatrix},
  \quad
  F =
  \begin{bmatrix}
    I &  0 &  \cdots & 0
  \end{bmatrix},
  \quad
  L =
    \begin{bmatrix}
      0 &  0 &  \cdots & I
  \end{bmatrix}.
\end{equation}
It is well known \cite{GohLanRod82} that $\lambda$ is a nonlinear
eigenvalue of the polynomial eigenproblem~\eqref{eq:polyeigprob} of
algebraic multiplicity $k$ if and only if it is a linear eigenvalue of
algebraic multiplicity $k$ of the linearization~\eqref{eq:6}.
Hence researchers \cite{TisMee01} have pursued the computation of
polynomial eigenvalues by standard eigensolvers applied to the linear
eigenproblem~\eqref{eq:6}. To do so using the 
FEAST algorithm,
the connection between the eigenspaces
of~\eqref{eq:6} and~\eqref{eq:polyeigprob} must be made precise, as
done in
Theorem~\ref{thm:projlin} below using $F$ and $L$.
Let us first describe the ingredients
of the algorithm  applied to the
linearization.

Replacing $A, B$ by $\Ac, \Bc,$ respectively,
in~\eqref{eq:SpectralProjection}
and~\eqref{eq:SN} we define $\Sc, \Sct, \Sc_N,$ and $\Sct_N$:
\begin{equation}
  \label{eq:15-big}
  \begin{aligned}
  \Sc & =
  \frac{1}{2\pi\ii}\oint_\Gamma (z \Bc - \Ac)^{-1} \Bc \, dz,
  &&
  \Sct =
  \frac{1}{2\pi\ii}\oint_{\Gamma} (z \Bc - \Ac)^{-*} \Bc^* \, dz, 
  \\  
    \Sc_N
    & = \sum_{k=0}^{N-1} w_k (z_k \Bc - \Ac)^{-1}  \Bc, 
    &&
    \Sct_N = \sum_{k=0}^{N-1} \bar{w}_k (z_k \Bc - \Ac)^{-*}  \Bc^*.
  \end{aligned}
\end{equation}
Given initial right and left subspaces
$\Ec_0, \Ect_0 \subset \C^{nd}$, the  FEAST algorithm, as
written out in~\eqref{eq:5}, computes a sequence of subspaces
$\Ec_\ell, \Ect_\ell$ by
\begin{equation}
  \label{eq:5big}
  \Ec_\ell = \Sc_N \Ec_{\ell-1},
  \qquad
  \Ect_\ell = \Sct_N \Ect_{\ell-1}
  \qquad \text{ for } \ell = 1, 2, \ldots.  
\end{equation}
In analogy with $E$ and $\Et$, we denote the right and left algebraic
eigenspaces of $z\Bc - \Ac$ associated to its (linear) eigenvalues
enclosed within $\vG$ by $\Ec$ and $\Ect$, respectively. Of course,
they are~\cite{Kato95}, respectively, the ranges of the Riesz
projections $\Sc$ and $\Sct$.  The relationships between these spaces
and the algebraic eigenspaces of the nonlinear $P(z)$ are given in the
next result, which can be concluded from well known results on matrix
polynomials.  We give a self-contained proof in
Section~\ref{sec:proof}. 

\begin{theorem}
  \label{thm:projlin}
  Let $E$ and $\Et$ be the right and left algebraic eigenspaces of the
  nonlinear eigenvalues of $P(z)$ enclosed in $\vG$, respectively.
  Then
  \begin{enumerate}
  \item $E = F \Ec$,
  \item $\Et = L\Ect$.
  \end{enumerate}
\end{theorem}

In view of Theorem~\ref{thm:projlin}, when the FEAST
algorithm~\eqref{eq:5big} converges to
$\Ec, \Ect$, mere truncation by $F$ and $L$ is guaranteed to yield the
algebraic eigenspaces needed in Problem~\ref{prb:nonlin}.

\begin{remark}  
  FEAST algorithms employing other contour integrals that can provably
  recover the wanted spaces $E, \Et$ (like in
  Theorem~\ref{thm:projlin}) are worthy of pursuit. To indicate why
  this might not be trivial, consider
  \[
    S_1= \frac{1}{2\pi \ii}
    \oint_\vG P(z)^{-1}\, dz.
  \]
  Even if it might
  appear to be a reasonable nonlinear generalization of the linear
  resolvent integral, for $P(z) = (z^2 - 1) A$ with any invertible
  $A \in \C^{n\times n}$, one can easily verify that $S_1=0$ when
  $\vG$ encloses both the nonlinear eigenvalues $\pm 1$ of $P(z)$.
  See also $S_2$ in Remark~\ref{rem:compare}.
\end{remark}

\subsection{An algorithm for solving polynomial eigenproblems}
\label{ssec:polyfeast}

In this subsection, we describe an efficient implementation
of~\eqref{eq:5big}. Implementing~\eqref{eq:5big} as stated would
require the inversion of $N$ linear systems of size $n d \times nd$,
significantly larger than the size of the $n \times n$ matrix
polynomial $P(\lambda)$. For large $nd$, due to the fill-in of sparse
factorizations, applying and storing $(z_k \Bc - \Ac)^{-1}$ at each
quadrature point $z_k$ becomes very expensive.  This drawback is
particularly serious for our application in
Section~\ref{sec:microfiber}, where as we shall see, each $A_i$ is
given as a large sparse matrix with $n \approx 10^7$.  Therefore, we
propose an implementation requiring only the inversion (or sparse
factorization) of $n \times n$ matrices (rather than  $nd \times nd$
matrices) at each quadrature point $z_k$, using the next result.

\begin{theorem}
  \label{thm:resolvent}
  Suppose $P(z)$ is invertible at some $z \in \C$ and consider
  $X, Y, W \in \C^{nd}$ block partitioned as in~\eqref{eq:19}.  Then
  the following identities hold.
  \begin{enumerate}
  \item The block components of $X = (z\Bc - \Ac)^{-1} Y$ are
    given by
  \begin{subequations}
    \label{eq:Xcomps}
    \begin{align}
      \label{eq:Xcomps-0}
          X_0
      & = P(z)^{-1}
        \left(
        - Y_{d-1} - A_d Y_{d-1}
        + \sum_{i=1}^d A_i \sum_{j=0}^{i-1} z^{i-1-j} Y_j         
        \right)
      \\ \label{eq:Xcomps-i}
      X_i & = z X_{i-1} - Y_{i-1}, \qquad i=1, 2, \ldots, d-1.
    \end{align}
  \end{subequations}
  \item  The block components of $\Xt = (z\Bc - \Ac)^{-*} W$ are given
    by
    \begin{subequations}
      \label{eq:Xcomps-adj}
      \begin{align}
        \label{eq:Xcomps:d-1}
        {\Xt}_{d-1}
        & = -P(z)^{-*} 
          \sum_{j=0}^{d-1} \bar{z}^j W_j,
          \quad
          {\Xt}_{d-2}
         = -W_{d-1} - \bar z A_d^* {\Xt}_{d-1} - A_{d-1}^*{\Xt}_{d-1},
        \\        \label{eq:Xcomps:i}
        {\Xt}_{i}
        & = -W_{i+1} + \bar z {\Xt}_{i+1} - A_{i+1}^* {\Xt}_{d-1},\qquad
          i=0, 1, \ldots, d-3.
      \end{align}
    \end{subequations}
  \end{enumerate}
\end{theorem}

\begin{alg}
  \caption{Polynomial FEAST Eigensolver for Problem~\ref{prb:nonlin}}
  \label{alg:polyfeast}
  \parbox{\textwidth}{Input contour $\Gamma$, quadrature $z_k, w_k$,  sparse coefficient
    matrices $A_0, \ldots, A_{d-1}, A_d\in \C^{n \times n}$, initial
    right and left eigenvector iterates given as columns of
    $Y, \Yt\in \C^{nd \times m}$, respectively, block partitioned as
    in~\eqref{eq:19} into $Y_j, \Yt_j \in \C^{n \times m}$, and
    tolerance $\veps>0$.}
\begin{pseudo}[kw]
  setup \\+
  \tn{Prepare $P(z_k)^{-1}$ by sparse factorization at each
    quadrature point~$z_k$.}
  \\-
  repeat \\+
  \label{algstep:inloop-1}
  \tn{Set all entries of workspace $\Rt, R \in \C^{nd \times m}$ to $0$.}
  \\
  for \tn{each $z_k$,\quad $k=0, \ldots, N-1$, do:}\\+
  \tn{Compute block components of $X\in \C^{nd\times m}$:}
  \\+ \label{algstep:Sappl0}
  \tn{
    $\displaystyle{
      X_0 \leftarrow 
       P(z_k)^{-1}\sum_{i=1}^d\sum_{j=0}^{i-1}
        z_k^{i-1-j}  A_i Y_j,
     }$
  }
  \\ \label{algstep:Sappl1}
  \tn{
    for $i=1, \ldots, d-1$ do: \;
     $\displaystyle{
      X_i \leftarrow 
      z_k X_{i-1} - Y_{i-1}
    }$.}
  \\-
  \tn{Increment $R \mathrel{+}= w_k X$.}
  \\
  \tn{Compute block components of $\Xt\in \C^{nd\times m}$:}
  \\+  \label{algstep:Stappl-start}
  \tn{
    $\displaystyle{
      \Xt_{d-1} \leftarrow 
      P(z_k)^{-*}
      \sum_{j=0}^{d-1}
      \bar z_k^{j} \Yt_j,
    }$
  }
  \\
  \tn{
    $\displaystyle{
      \Xt_{d-2}
      \leftarrow
      -\Yt_{d-1}
      - \bar z_k A_d^* \Xt_{d-1}
      - A_{d-1}^* \Xt_{d-1}},$
  }
  \\ \label{algstep:Sappl-end}
  \tn{
    for $i=d-3, \ldots, 1, 0,$ \;do: \; 
    $\displaystyle{
      \Xt_{i}
      \leftarrow
      A_{i+1}^* \Xt_{d-1} -\Yt_{i+1} + \bar z_k \Xt_{i+1}.
    }$ \quad 
  }
  \\-
  \tn{Increment $\Rt \mathrel{+}= \bar w_k \Xt$.}
  \\-
  endfor
  \\ \label{algstep:kerclean-begin}
  \tn{$G \leftarrow \Rt^* \Bc R$.}
  \\
  \tn{Compute biorthogonal $V, \Vt\in \C^{m\times m}$ such that 
    $\Vt^* G V = \diag(d_1, \ldots, d_m).$}
  \\
  \tn{$Y \leftarrow R V$, \; $\Yt \leftarrow \Rt \Vt$.}
  \\ 
  for \tn{$\ell=1, \ldots, m$ do:}
  \\+
  \tn{If $d_\ell \approx 0$: then remove $\ell$th columns of $\Yt$ and
    $Y$,}
  \\
  \tn{else: rescale $\ell$th column of $\Yt$ and $Y$ by $|d_\ell|^{-1/2}$.}
  \\- 
  endfor    \label{algstep:kerclean-end}
  \\ \label{algstep:Ritz}
  \tn{Assemble small Ritz system:
    $A_Y \leftarrow  \Yt^* \Ac Y$, \quad
    $B_Y \leftarrow \Yt^*\Bc Y$.}
  \\
  \label{algstep:denseeig}
  \tn{Compute Ritz values 
    $\vL = \diag(\lambda_1, \ldots, \lambda_m)$ and
    $W, \Wt \in \C^{m\times m}$  satisfying }\\*
  & \tn{$\Wt^* A_Y W = \vL$, \quad
    $\Wt^* B_Y W = I$.}
  \\
  \label{algstep:Yupdate}
  \tn{$Y \leftarrow Y W$, \quad $\Yt \leftarrow \Yt \Wt$.}
  \\
  \tn{Periodically check:}
  \tn{if $\lambda_\ell$ falls outside $\vG$,
    remove  $\ell$th columns of $Y$ and $\Yt$.}
  \\-
  until \tn{maximal difference of successive $\vL$ iterates is
    less than $\veps$.}
  \\
  output \tn{eigenvalue cluster $\{\lambda_\ell\}$,
    left \& right eigenvectors in
    columns of $L\Yt$ and $FY$.}
\end{pseudo}
\end{alg}

\bigskip

Theorem~\ref{thm:resolvent} is proved in Section~\ref{sec:proof}.  A
FEAST implementation based on it is given in
Algorithm~\ref{alg:polyfeast}, which we now describe.  The algorithm
is written with small $(m)$ eigenvalue clusters and large $(n)$ sparse
$A_i$ in mind ($m \ll n$).  We also have in mind semisimple
eigenvalues, since we want to use standard software tools for small
diagonalizations (avoiding the complex issue of stable computation of
generalized eigenvectors).  Computation of $\Sc_N Y$ and $\Sct_N \Yt$
occur in steps~\ref{algstep:Sappl0}--\ref{algstep:Sappl1}
and~\ref{algstep:Stappl-start}--\ref{algstep:Sappl-end} of
Algorithm~\ref{alg:polyfeast}, via the identities of
Theorem~\ref{thm:resolvent}. After the computation of $\Sc_N Y$
and $\Sct_N \Yt$, the algorithm assembles a small ($m \times m$) Ritz
system, based on the new eigenspace iterate, in
step~\ref{algstep:Ritz}. Subsequently, step~\ref{algstep:denseeig}
attempts to diagonalize this.  In practice, one must also handle
exceptions in the event this diagonalization fails due to a (close to)
defective eigenvalue, details which we have omitted from
Algorithm~\ref{alg:polyfeast}, since we did not need them in our
application.

Recall that we do not require $\Bc$ to be invertible. 
Lines~\ref{algstep:kerclean-begin}--\ref{algstep:kerclean-end}
of the algorithm 
remove vectors in $\Kc = \ker \Bc$, the null space of
$\Bc$, from the iteration.  It is immediate from~\eqref{eq:15-big} that
$\Kc$ is contained in the eigenspaces of $\Sc$ and $\Sc_N$ associated
to their zero eigenvalue. Since these operators have their dominant
eigenvalue away from zero, the subspace iteration~\eqref{eq:5big} will
{\em filter out elements of $\Kc$} from its iterates.  
Let $K = \ker A_d$. Note that 
$\Kc = L' \,K$.  From the definitions in Subsection~\ref{ssec:polyneigen}, it is
obvious that any nontrivial element of $K$ is an eigenvector
of~\eqref{eq:polyeigprob} corresponding to eigenvalue $\infty.$
Therefore, $\Kc$ being filtered out amounts to filtering out the
eigenspace of $\lambda=\infty$.


In the optics applications we are about to consider in the next two
sections, $A_d$ turns out to be Hermitian and negative
semidefinite. Then, $\Bc = \Bc^*$ is positive semidefinite and
$(x, y)_\Bc = y^* \Bc x$ defines a semi-inner product (and an inner
product on $\Kc^\perp$). Moreover, it is easy to see that
\begin{equation}
  \label{eq:21}
  (\Sc x, y)_\Bc = (x, \Sct y)_\Bc, \qquad 
  (\Sc_N x, y)_\Bc = (x, \Sct_Ny)_\Bc  
\end{equation}
for all $x, y \in \C^{nd}$, i.e., $\Sct$ and $\Sct_N$ are the
$\Bc$-adjoints of $\Sc$ and $\Sc_N$, respectively.  The first equation
of~\eqref{eq:21} implies that $(\Ec, \ker \Sct)_\Bc=0$. Hence the
wanted right eigenfunctions (in $\Ec$)
are $\Bc$-orthogonal to the
unwanted left ones (in $\ker \Sct$), and vice versa,
since we also have
$(\ker \Sc, \Ect)_\Bc = 0$.  When the iterates $Y, \Yt$ of
Algorithm~\ref{alg:polyfeast} converge, their respective column spaces
inherit these orthogonality properties.
Note also that after the update in step~\ref{algstep:Yupdate}, the
columns of the iterates $Y$ and $\Yt$ are
$\Bc$-biorthogonal, i.e., $\Yt^* \Bc Y = I$.

\begin{remark}
  \label{rem:compare}
  Using Theorem~\ref{thm:resolvent}'s \eqref{eq:Xcomps-0}, it is easy to
  see that the contour integral $\Sc$ satisfies
  \begin{equation}
    \label{eq:FSY}
    F \Sc Y  = \frac{1}{2\pi\ii} \oint_\vG P(z)^{-1} \sum_{i=1}^d A_i
    \sum_{j=0}^{i-1} z^{i-1-j} Y_j\, dz, \qquad Y \in \C^{nd \times m}.
  \end{equation}
  By Theorem~\ref{thm:projlin}, the range of $F\Sc$ satisfies
  $\ran(F\Sc) = E$, which was the basis for correctness of
  Algorithm~\ref{alg:polyfeast}.
  Another interesting application of
  Theorem~\ref{thm:resolvent} is in analyzing 
  the algorithm of~\cite{GavMiePol18}, which 
  is based on another contour integral map
  $S_2:\C^{n\times m} \to \C^{n \times m}$, defined for some
  $\mu_i \in \C$ and $y=[y_1, \ldots, y_m] \in \C^{n \times m}$
  (with $y_k \in \C^n$), by
  \[
    \Sm{\mu_k}_2 = \frac{1}{2\pi \ii} \oint_\vG P(z)^{-1}
    \frac{P(z) - P(\mu_k)}{z - \mu_k}\, dz,
    \qquad
    S_2 y = [\Sm{\mu_1}_2 y_1,\ldots, \Sm{\mu_m}_2 y_m ].
  \]
  They seek eigenvector approximations from a different
  space~$R_2 = \sum_{k=1}^m
  \ran{(\Sm{\mu_k}_2)}$. 
  Factoring $z-\mu_k$ out of $P(z) - P(\mu_k)$, we find that 
  \[
    \Sm{\mu_k}_2 y= \frac{1}{2\pi\ii} \oint_\vG
    P(z)^{-1} \sum_{i=1}^d A_i \sum_{j=0}^{i-1} z^{i-1-j} \mu_k^j y\, dz.
  \]
  Comparing with~\eqref{eq:FSY} and choosing $Y_j = \mu_k^j y$, we
  establish that $R_2 \subseteq \ran(F\Sc) = E$. The reverse
  inclusion does not always hold. For example, if $P(z) = \left[
    \begin{smallmatrix}
      1 & z\\
      1 & z^2 
    \end{smallmatrix}
  \right] \in \C^{2\times 2},$ $\mu_1 = \mu_2$, and $\vG$
  encloses both eigenvalues $0, 1$ of $P(z)$, then $R_2$
  is the one-dimensional space spanned by
  $ \left[\begin{smallmatrix} -\mu_1 \\ 1
    \end{smallmatrix}\right],$ while the exact eigenspace $E$ is the
  span of $e_1 = 
  \left[\begin{smallmatrix} 1 \\ 0 
    \end{smallmatrix}\right]$
  and $e_2 =  \left[\begin{smallmatrix} 0 \\1 
    \end{smallmatrix}\right]$. 
  This example also shows that dimensions of $E$ may be lost
  even when applying $S_2$ to a basis of $E$:
  $S_2[e_1, e_2] = [\Sm{\mu_1}_2 e_1, \Sm{\mu_2}_2 e_2] = 
  \left[\begin{smallmatrix} 
      0  & -\mu_2 \\ 
      0 & 1
    \end{smallmatrix}\right]$ for any $\mu_1, \mu_2$.
\end{remark}

\section{Leaky modes of optical fibers}
\label{sec:leaky}

Assuming that the material properties of an optical fiber do not vary
in the longitudinal ($x_3$) direction, we consider the plane (in
$x_1, x_2$ coordinates) of its transverse cross section.  The
refractive index can be modeled as the piecewise function on the
transverse plane,
\begin{subequations}
  \label{eq:fibermodel}
\begin{align}
	\label{eq:12}
	n(x_1, x_2) &
	=
	\begin{cases}
		n_1(x_1, x_2), & r \le R_0,\\
		n_0, & r >R_0,	\\
	\end{cases}
\end{align}
where $r = \sqrt{x_1^2 + x_2^2}$, $R_0>0$ is a radius beyond which the
medium is homogeneous, $n_0$ is a constant representing the refractive
index of the homogeneous medium, and $n_1$ is the given refractive
index of the fiber. The vector Maxwell system for time-harmonic light
propagation is often simplified to a scalar equation (see
\cite{Reide16}, or to see the specific assumptions in this process,
see e.g.,~\cite{DrakeGopalGoswa20}) when computing fiber modes.
Accordingly,
a {\em transverse mode} of the optical fiber is represented by a
nontrivial scalar field $u: \R^2 \to \C$, with an accompanying {\em
  propagation constant} $\beta \in \C$. Together they satisfy the
Helmholtz equation
\begin{align}
	\label{eq:11}
  \Delta u + k^2 n^2 u & = \beta^2 u
                         \qquad \text{ in } \R^2
\end{align}
\end{subequations}
for some wavenumber $k \in \R$ given by the operating frequency of the
fiber.  An isolated real value of $\beta$ and an associated mode
function $u$ that decays exponentially as $r \to \infty$ are usually
referred to as guided modes of the fiber. Here, we are concerned with
computation of complex isolated propagation constants $\beta$ and the
corresponding outgoing field~$u$.  In the optics literature, such
modes are usually referred to as {\em leaky modes} \cite{Marcu91}, but
they are also known by other names such as resonances or quasi-normal
modes~\cite{GopalMoskoSanto08, KimPasci09}. To compute such modes, we
truncate the infinite domain after using a perfectly matched layer
(PML) \cite{Beren94}. Venturing out of the approach of~\cite{Beren94}
and viewing PML as a complex coordinate change resulted in better
understanding of PML~\cite{ChewWeedo94, ColliMonk98, KimPasci09}. In
this section, we follow the recent approach of~\cite{NannenWess18} and
apply the previously described FEAST algorithm to their PML
discretization.

Since the wavenumbers in the optical regime are high and transverse
dimensions of optical fibers are several orders smaller, it is
important to nondimensionalize before discretization. Let $L$ denote a
fixed characteristic length scale for the transverse dimensions of the
fiber.  Then, in the nondimensional variables
\begin{equation}
  \label{eq:23}
  \hat x_1 = \frac{x_1}{L},\quad \hat x_2 = \frac{x_2}{L}, 
\end{equation}
the function $\hat u(\hat x_1, \hat x_2) = u(L \hat x_1, L \hat x_2)$
satisfies $\hat \Delta \hat u = L^{-2} \Delta u$ where
$\hat \Delta = \d^2 / \d \hat x_1^2 + \d^2 / \d \hat x_2^2$, 
and equation~\eqref{eq:11} transforms into
\[
  \hat \Delta \hat u + L^2(k^2 \hat n^2 - \beta^2) \hat u = 0. 
\]
where $\hat n(\hat x_1, \hat x_2) = n(\hat x_1 L, \hat x_2 L)$.  Let 
\[
  Z^2 = L^2(k^2 n_0^2 - \beta^2),
  \qquad
  V(\hat x_1, \hat x_2) = L^2k^2( n_0^2 - \hat n^2).
\]
Clearly, the function $V$ is supported only in the region
$\hat r \le \Rh_0$ where $\Rh_0 = R_0/L$ and $\hat r = r/L$. The
problem of finding a leaky mode pair $u, \beta$ has now become the
problem of finding a constant $Z$ and an associated nontrivial $\hat u$
satisfying
\begin{subequations}
    \label{eq:non-dim-form}
    \begin{align}
      \label{eq:non-dim-form-pde}
      -\hat \Delta \hat u + V \hat u & = Z^2 \hat u, && \quad
    \text{ in } \R^2,
    \\ \label{eq:non-dim-form-bc}
    \hat u \text{ is} & \text{ outgoing,}
    && \quad \text{ as } \hat r \to \infty.
  \end{align}    
\end{subequations}
This form, in addition to being nondimensional, facilitates comparison
with the mathematical physics literature where the spectrum of
$-\Delta + V$ is extensively studied for various ``potential wells'' $V$.

The condition at infinity in~\eqref{eq:non-dim-form-bc} should be
satisfied by the solution in the unbounded region $\rh >
\Rh_0$. There, since $V$ vanishes,
equation~\eqref{eq:non-dim-form-pde} takes the form
\begin{equation}
  \label{eq:25}
\hat \Delta \hat u + Z^2 \hat u = 0, \qquad \rh > \Rh_0.  
\end{equation}
For {\em real} values of $Z$ in this equation, the boundary
condition~\eqref{eq:non-dim-form-bc} is easily realized by the
Sommerfeld radiation condition
$\lim_{\rh \to \infty} \sqrt{\rh} (\d_{\rh} \hat u - \ii Z \hat u) =
0$, which selects outgoing waves. Moreover, in this case, the general
solution in the $\rh > \Rh_0$ region can be derived using separation of
variables:
\begin{equation}
  \label{eq:20}
  \hat u(\rh, \theta)
  = \sum_{\ell=-\infty}^\infty c_\ell H_\ell^{(1)}(Z \rh) e^{\ii \ell
    \theta},
  \qquad \rh > \Rh_0
\end{equation}
for some coefficients $c_\ell$. Here $(\rh, \theta)$ denotes polar
coordinates and $H_\ell^{(1)}$ denotes the $\ell$th Hankel function of
the first kind.  For {\em complex} $Z$, a simple prescription of the
boundary condition~\eqref{eq:non-dim-form-bc} that $\hat u$ ``is
outgoing'' is the requirement that $\hat u$ have the same
form~\eqref{eq:20} even when $Z$ is complex, using the analytic
continuation of the Hankel function from the positive real line.  The
resonances we are interested in computing will have nondimensional
$Z$-values below the real line. (Note that these locations are
different from the locations of the optical propagation constants
$\beta$.)  The well-known~\cite{AbramStegun72} asymptotic behavior of
the Hankel function,
\begin{equation}
  \label{eq:Hankel-asym}
  H^{(1)}_\ell(\zeta) \sim \kappa_\ell\frac{e^{\ii \zeta}}{\zeta^{1/2}},
    \qquad |\zeta|\to \infty, -\pi < \arg \zeta< 2\pi,     
\end{equation}
with $\kappa_\ell = (2/\pi)^{1/2} e^{-\ii(\ell\pi/2 +\pi/4)}$, 
tells us that
when the imaginary part  $\Im(Z) < 0$,
the summands in~\eqref{eq:20}
 blow up exponentially at infinity and hence $\uh$  generally cannot be
$L^2$-normalized (the reason for the name quasi-normal mode).

\subsection{Discretization based on PML}

When interpreted as a complex coordinate change, PML maps the
coordinates $\hat x = (\hat x_1, \hat x_2) \in \R^2$ to
$\xt = (\xt_1, \xt_2) \in \C^2$ using a transformation of the form
\begin{equation}
  \label{eq:24}
  \begin{pmatrix}
    \xt_1 \\ \xt_2
  \end{pmatrix}
  =
  \frac{\eta(\hat r)}{ \hat r}
  \begin{pmatrix}
    \hat x_1 \\  \hat x_2
  \end{pmatrix}
\end{equation}
for some $\eta$ with the property that $\eta(\rh) = \rh$ for
$\rh \le \Rh$ for some $\Rh> \Rh_0$, i.e., the PML starts at $\Rh > \Rh_0$
and leaves the $\rh \le \Rh$ region untouched. Consider what happens
to the solution expression~\eqref{eq:20} under this change of
variable. Substituting $\rt = (\xt_1^2 + \xt_2^2)^{1/2} = \eta(\rh)$
for $\rh$ in~\eqref{eq:20}, we find that the summands now have the
term $H_\ell^{(1)}(Z \rt)= H_\ell^{(1)}(Z \eta(\rh))$ whose asymptotic
behavior for large arguments, per~\eqref{eq:Hankel-asym},
imply that they decay exponentially whenever there are
$c, \rh_1>0$ such that 
\begin{equation}
  \label{eq:22}
  \Im(Z \eta(\rh)) > c\, \hat r, \qquad \rh > \rh_1.
\end{equation}
When~\eqref{eq:22} holds, the leaky mode with exponential blow up is
transformed to a function with exponential decay at infinity.  PML
exploits such an exponential decay to truncate the infinite domain and
impose zero Dirichlet boundary conditions at an artificial boundary
where the transformed solution is close to zero. Such complex
transformations have been effectively utilized in the mathematical
literature of resonances~\cite{Simon73} decades before the term PML
was coined.

Various choices of $\eta$ were proposed in the literature.  For
example, in~\cite{ColliMonk98}, one finds the choice
$\eta(\rh) = \rh + \ii Z^{-1} \varphi(\rh)$, for some non-negative
function $\varphi$, applied to treat the source problem analogous
to~\eqref{eq:25} where $Z$ is viewed as a given wavenumber or
``frequency''. This is a ``frequency dependent'' complex change of
coordinates.  The authors of~\cite{KimPasci09} used another $C^2$
function $0 \le \varphi \le 1$ and set
$\eta(\rh) = \rh + \ii \alpha \rh\varphi(\rh)$ for some constant
$\alpha>0$.  When applied to the eigenproblem~\eqref{eq:non-dim-form},
their choice is independent of the frequency $Z$ (which is now an
unknown eigenvalue), and therefore has the advantage of resulting in a linear
eigenproblem (for $Z^2$). However, since
$\Im (Z \eta(\rh)) = (\Im Z + \alpha \varphi(\rh) \Re Z) \rh, $ the
condition \eqref{eq:22} for exponential decay (while satisfied for
some choices of $\alpha$ in relation to $Z$) is not satisfied by their
choice unconditionally.  This creates practical difficulties in
separating spurious modes from real ones.

Hence a return to a
frequency-dependent choice was advocated in~\cite{NannenWess18, Wess20},
notwithstanding the complication that a $Z$ dependence in $\eta$ would
lead to a nonlinear eigenproblem.  We adopt their choice in our
computations and set
\[
  \eta(\rh) = 
  \left\{
    \begin{aligned}
      &\rh,
      && \quad \rh \le \Rh,
      \\
      &\frac{1+\ii\alpha}{Z}(\rh - \Rh) + \Rh,
      && \quad \rh > \Rh.
    \end{aligned}
  \right.
\]
Then
\begin{equation}
  \label{eq:30}
\Im( Z \eta) = \alpha (\rh - \Rh) + \Rh \Im Z,   
\end{equation}
so~\eqref{eq:22} holds for any $\alpha>0$ and any $Z$
by taking $\rh$ 
large enough.

The mapped eigenfunction $\ut(\rh, \theta) = \uh(\eta(\rh), \theta)$
is approximated in the (complex valued) Sobolev space $H^1(\om)$ on
the finite domain $\om = \{ (\hat r, \theta): \rh < \Rout\}$ where
$\Rout$ is to be chosen large enough. In our computations using the
contour integral solver, the given contour $\vG$ determines the
minimal imaginary part of a potential eigenvalue to be found, which
can be used within~\eqref{eq:30} and~\eqref{eq:Hankel-asym} to ({\it a
  priori}) estimate a distance $\Rout$ that gives a desired decay.
The mapped function $\ut$ satisfies the following variational
formulation: find $\ut \in H^1(\om)$ satisfying
\begin{equation}
  \label{eq:26}
  \int_\om a(\hat x) \nabla \ut \cdot \nabla v
  \, d\hat x
  +
  \int_\om V \ut v\, (\det J) \, d\hat x
  = Z^2 \int_\om \ut v \, (\det J) \, d\hat x
\end{equation}
for all $v \in H^1(\om)$ where $a = (\det J) J^{-1} [J']^{-1}$ and
$J_{ij} = \d \xt_i/\d \hat x_j$ denotes the Jacobian of the complex
mapping. Equation~\eqref{eq:26} is derived by applying the complex
change of variable $r \mapsto \rt$ to~\eqref{eq:non-dim-form-pde},
then multiplying by a test function $v \in H^1(\om)$, integrating by
parts, and using the boundary condition $a \nabla \ut \cdot n =0$ on
$\d\om$. One may also use the Dirichlet boundary conditions (setting
the weak form in $\Ho^1(\om)$ instead) due to the exponential decay
within PML, but using the natural boundary condition allows an
implementation to test ({\it a posteriori}) whether the computed
solution has actually decayed in size at $\d\om$ (and if it has not,
increase $\Rout$ further and recompute).  Tracking the dependence
of $Z$ in each integrand and simplifying, one sees that~\eqref{eq:26}
yields a rational eigenproblem for~$Z$.  As pointed out
in~\cite{NannenWess18, Wess20}, further simplifications are possible by a
judicious choice of test functions, as described next.

Replacing $v$ in~\eqref{eq:26} by $\vt = v \eta(\rh) / \Rh,$ the first
integrand can be simplified to
\begin{equation}
  \label{eq:16auv}
  a(\hat x) \nabla \ut \cdot \nabla \vt
  =
  \frac{\dot{\eta} \rh}{\Rh} \nabla \ut \cdot \nabla v
  + \frac{1}{\Rh}
  \bigg(
  \frac{\eta^2}{\dot{\eta} \rh^3} -
  \frac{\dot{\eta}}{\rh}
  \bigg)
  \xh \xh' \nabla \ut \cdot \nabla v
  +
  \frac{\eta}{\rh^2 \Rh} \nabla \ut \cdot \xh v  
\end{equation}
where $\dot{\eta} = d \eta / d\rh$. Examining the dependence of each
term on $Z$ through $\eta$ and $\dot\eta$, we find that it is now
possible to write this integrand as a polynomial in $Z$, save for a
common factor of $Z^{-1}$. Since we are not interested in the $Z=0$
case as an eigenvalue, we multiply through by $Z$ to get
\[
  \int_\om Za(\hat x) \nabla \ut \cdot \nabla \vt
  \, d\hat x
  +
  \int_\om Z V \ut \vt\, (\det J) \, d\hat x
  = Z^3 \int_\om \ut \vt \, (\det J) \, d\hat x
\]
and express it as a polynomial in $Z$:
\begin{equation}
  \label{eq:27}
    \sum_{i=0}^3 Z^i\, b_i(\ut, v) = 0,
\end{equation}
where  
\begin{align*}
  b_0(w, v)
  & = (1+ \ii\alpha)
    \int_\ompml
    \bigg[
    \frac{\rh }{\Rh} \nabla w \cdot \nabla v +
    \bigg(
    \frac{(\rh - \Rh)^2}{\rh^3} - \frac 1 {\rh}
    \bigg)
    \frac{\xh \xh'}{\Rh}
    \nabla w \cdot \nabla v
    \bigg] d \xh
  \\
  & + (1+ \ii\alpha)
  \int_\ompml
  \frac{\rh - \Rh}{\Rh \rh^2} \nabla u \cdot \xh v
    \,d \xh
    - (1+ \ii\alpha)^3 \int_\ompml \frac{(\rh - \Rh)^2}{\Rh \rh} w v
    \, d\xh,
    \\
  b_1(w, v)
  & = \int_\omint\big( \nabla w \cdot \nabla v +V wv\big) \, d\xh
    +
    \int_\ompml
    \bigg[
    \frac{2(\rh - \Rh)}{\rh^3} \xh \xh' \nabla w \cdot \nabla v
    + \nabla w \cdot \frac{\xh}{\rh^2} v
    \bigg]
    d\xh
  \\
  & -2(1+ \ii \alpha)^2\int_\ompml \frac{\rh - \Rh}{\rh} w v \, d\xh,
  \\
  b_2(w, v)
  & = \frac{\Rh}{1+\ii\alpha} \int_\ompml \frac{\xh\xh'}{\rh^3} \nabla
    w \cdot \nabla v \, d\xh
    - 
    {\Rh}(1+\ii\alpha) \int_\ompml \frac{1}{\rh} w v \, d \xh,
  \\
  b_3(w, v)
  & = -\int_{\omint} w v \; d\xh,
\end{align*}
and
$\ompml = \{ (\rh, \theta) \in \om: \rh > \Rh\}$ and $\omint = \om
\setminus \ompml$.

Finally, to discretize~\eqref{eq:27}, we use a geometrically
conforming triangular finite element mesh $\oh$ and the Lagrange
finite element space $W_{hp} = \{ v \in H^1(\om): \; v|_K$ is a
polynomial of degree at most $p$ in each mesh element $K \in
\oh\}$. Here $h = \max_{K \in \oh} \diam K$ is the mesh size
parameter. We seek a nontrivial $\ut_{hp} \in W_{hp}$ together with a
$Z \in \C$ satisfying
\begin{subequations}
  \label{eq:discreteCubicEVP}
  \begin{equation}
  \label{eq:28}
  \sum_{i=0}^3 Z^i\, b_i(\ut_{hp}, v) = 0 \qquad \text{ for all } v
  \in W_{hp}.
  \end{equation}
Letting $\{\phi_j: j=1, \ldots, n\}$ denote a finite element basis for
$W_{hp}$, we define matrices $A_i \in \C^{n \times n}$ by
\begin{equation}
  \label{eq:29}
    [A_i]_{kl}  = b_i(\phi_l, \phi_k).
\end{equation}
Then, expanding $\ut_{hp}$ in the same basis,
$\ut_{hp} = \sum_{j=1}^n c_j \phi_j$ for some $c \in \C^n$,
equation~\eqref{eq:28} yields the cubic eigenproblem
\begin{equation}
  \label{eq:31}
 P(Z) c =  \sum_{i=0}^3 Z^i A_i  c= 0 
\end{equation}
\end{subequations}
for the coefficient vector $c$ of $\ut_{hp}$.

The eigenproblem~\eqref{eq:discreteCubicEVP} is clearly a problem of
the form~\eqref{eq:polyeigprob} we considered in the previous
section. We shall solve it using Algorithm~\ref{alg:polyfeast} in the
following sections for specific fiber configurations. Note
that the $\Bc$ that arises from~\eqref{eq:discreteCubicEVP}
is not invertible: the matrix $A_3$ from~\eqref{eq:discreteCubicEVP}
that forms the last block of $\Bc$, is Hermitian, is negative
semidefinite, and has a large null space.  All finite element
functions in $W_{hp}$ that are supported in the $\ompml$ region are in
this null space and form one source of spurious modes in typical
resonance computations.  As noted in Subsection~\ref{ssec:polyfeast}, these
functions are associated to the eigenvalue $\infty$ and are  
automatically removed from the subspace iterates in
Algorithm~\ref{alg:polyfeast}, so spurious modes supported in the PML
region cannot pollute the results. This is a useful feature arising
from the combination of the frequency-dependent PML and
Algorithm~\ref{alg:polyfeast}.

\begin{remark}  
  Instead of~\eqref{eq:31}, it is possible to arrive at a cubic
  eigenproblem where all the matrices $A_i$, and hence $P(Z)$, are
  (complex) symmetric, an advantageous feature for some sparse
  factorization techniques. To do so, set
  $\tilde u = (\eta(\rh)/\hat R)^{1/2} \breve u (\hat x) $ and
  $v = (\eta(\rh)/\hat R)^{1/2} \breve v(\hat x)$
  in~\eqref{eq:26}. Then instead of \eqref{eq:16auv}, we obtain
  \[
    a(\hat x) \nabla \ut \cdot \nabla v =
    \frac{\dot{\eta} \rh}{\Rh}
    \nabla \breve u \cdot \nabla \breve v
    +
    \bigg(
    \frac{\eta^2}{\dot{\eta} \rh} -
    {\dot{\eta} \rh}
    \bigg)
    \frac{\xh \xh'}{\Rh}
    \nabla \breve u \cdot \nabla \breve v
    +
    \frac{\eta}{2\Rh\rh^2}(\nabla \breve u \cdot \xh \breve v
    + \breve u \nabla \breve v \cdot \xh)
    +
    \frac{\dot\eta}{4 \Rh\rh} \breve u \breve v,
  \]
  an integrand that is symmetric in $\breve u$ and $\breve v$.
\end{remark}

\subsection{Verification using a step-index fiber}
\label{ssec:verif}

The case of a step-index fiber provides an example for verifying
numerical methods for computing leaky modes. Leaky modes can be solved
 in closed form for step-index fibers. A step-index fiber is
modeled by a cylindrical core region of a constant refractive index
surrounded by a cladding region of a slightly lower constant
refractive index. Since the cladding diameter is usually many times
larger than the core diameter, the modes of the fiber can be
approximated using the problem~\eqref{eq:fibermodel} with $n_0$ set to
the cladding refractive index, $n_1$ set to the (constant) core index,
and $R_0$ set to the core radius. We nondimensionalize by
setting the length scale to the {\em core radius}
$
  L = R_0
$
and obtain~\eqref{eq:non-dim-form} with $\Rh_0=1$ and 
\[
  V = \left\{
    \begin{aligned}
      & -V_1^2,  && \rh \le 1\\
      & 0,  && \rh > 1,
    \end{aligned}
  \right.
\]
where $V_1^2 = R_0^2k^2(n_1^2 - n_0^2)$ is often called the {\em
  normalized frequency} (or sometimes,  the ``V-number'')
\cite{Reide16} of the step-index fiber.  The
fiber core region has now been transformed into the nondimensional
unit disk $\rh \le 1$. Using the standard interface transmission
conditions for the Helmholtz equation, we may
rewrite~\eqref{eq:non-dim-form} as the system
\begin{subequations}
  \label{eq:step-index-nondim}
  \begin{align}
    \label{eq:step-index-nondim-X-core}
    \hat\Delta \uh + X^2 \uh & = 0    && \rh < 1,
    \\
    \label{eq:step-index-nondim-Z-clad}
    \hat\Delta \uh + Z^2 \uh & = 0    && \rh > 1,
    \\
    \label{eq:step-index-nondim-interface}
  \jmp{\uh} = \jmp{\d\uh/\d\rh} & = 0 && \rh = 1,
\end{align}  
\end{subequations}
with $X^2 = V_1^2 + Z^2$.  In~\eqref{eq:step-index-nondim-interface},
the notation $\jmp{v}$ indicates the jump (defined up to a sign) of a
function $v$ across the core-cladding interface $\rh=1$.
We proceed to analytically solve for the
general form of solutions of the first two equations and then match
them by the third equation.

\begin{figure}
  \centering
  \begin{subfigure}[t]{0.5\textwidth}
    \begin{tikzpicture}
      \begin{groupplot}
        [
        width=0.75\textwidth,
        height=0.3\textheight,
        group style={
          group size=1 by 2,
          vertical sep=5pt,
        }
        ]

        \nextgroupplot
        [
        ylabel={$\Im(\beta)$},
        y label style={at=(ticklabel* cs:0.1), yshift=30},
        ymin=1000,
        xmin=8520000, xmax=8570000, 
        axis x line=top,
        xticklabels={,,},    
        xtick scale label code/.code={}, 
        axis y discontinuity=parallel,
        height=5cm,
        x axis line style=-,
        ytick={5000, 10000},
        yticklabels={5000, 10000},
        ytick scale label code/.code={}, 
        legend style={at={(1.3, 0.3)}, anchor=south, font=\footnotesize},
        ]
        
        \addplot
        [
        scatter,
        only marks,
        point meta=explicit symbolic,
        scatter/classes={
          g={mark=star, black},%
          l0={mark=*, red!50, draw=black},%
          l1={mark=halfcircle*, red, draw=black},
          l2={mark=otimes*, green, draw=black},
          l3={mark=oplus*, blue!40, draw=black}
        },
        ]
        table[x=x, y=y,  meta=label] {py/stepindexews.dat};

        \legend{
          guided,
          leaky ${l} =0$,
          leaky ${l} =1$,
          leaky ${l} =2$,
          leaky ${l} =3$
        }
        
        \nextgroupplot
        [
        ymin=-100,ymax=500,
        xmin=8520000, xmax=8570000, 
        ytick={0, 200, 400},
        minor ytick={100, 300},
        axis x line=bottom,
        xlabel={$\Re(\beta)$}, 
        height=3cm,
        x axis line style=-,
        yminorgrids=true,
        ymajorgrids=true,
        grid=both,
        ]
        
        \addplot    
        [
        scatter,
        only marks,
        point meta=explicit symbolic,
        scatter/classes={
          g={mark=star, black},%
          l0={mark=*, red!50, draw=black},%
          l1={mark=halfcircle*, red, draw=black},
          l2={mark=otimes*, green, draw=black},
          l3={mark=oplus*, blue!40, draw=black}
        },
        ]
        table[x=x, y=y,  meta=label] {py/stepindexews.dat};


      \end{groupplot}
    \end{tikzpicture}
    \subcaption{Physical propagation constants $\beta$}
    \label{subfig:physical}
  \end{subfigure}    
  \begin{subfigure}[t]{0.45\textwidth}
    \begin{tikzpicture}
      \begin{axis}
        [
        width=\textwidth,
        height=0.3\textheight,
        xlabel={$\Re(Z)$},
        ylabel={$\Im(Z)$},
        y label style={xshift=-40},
        ymax=1.2, ymin=-3,
        xmin=0,
        axis x line=middle,
        xticklabel style={anchor=south},
        xtick={0, 4,  8},
        axis y line=left,
        grid=both,
        ]

        \addplot
        [
        scatter,
        only marks,
        point meta=explicit symbolic,
        scatter/classes={
          g={mark=star, black},%
          l0={mark=*, red!50, draw=black},%
          l1={mark=halfcircle*, red, draw=black},
          l2={mark=otimes*, green, draw=black},
          l3={mark=oplus*, blue!40, draw=black}
        },
        every inner x axis line/.append style={-stealth},
        ]
        table[meta=label] {
          x                     y                     label
          5.351835174490589    -1.334942821742271     l0
          9.061909860015561    -1.6380419749437216    l0
          12.471958660800317   -1.8735366508309028    l0
          
          2.903324474874013    -1.1019639101931995    l1
          7.20192238202425     -1.492413782571426     l1
          10.73969426521824    -1.7614196413353773    l1

          0.3044351516911604   -1.0376997172549105    l2
          12.308875069330629   -1.8719715580460672    l2
          4.949838513024925    -1.2780715576847275    l2
          8.83440294025797     -1.6288332276862594    l2

          1.9577933269202556   -0.18543240054910468   l3
          6.586842343191748    -1.4379007252654379    l3
          10.348866263178596   -1.7522310599718856    l3
        };

        \draw[->, dotted]
        (1.9577933269202556, -0.18543240054910468)--(2.7, 0.6)
        node[above right]       
        {\tiny{
            \hspace{-1.2cm}
            (Fig.~\ref{subfig:intensities-step-index} shows the
            modes located here.)}};
      \end{axis}
    \end{tikzpicture}
    \subcaption{Nondimensional eigenvalues $Z$\label{subfig:Z}}
  \end{subfigure}
  \caption{Locations of nondimensional eigenvalues (right) yielding
    some fiber modes and their corresponding $\beta$-values
    (left).}
  \label{fig:ew-leaky-guided} 
\end{figure}
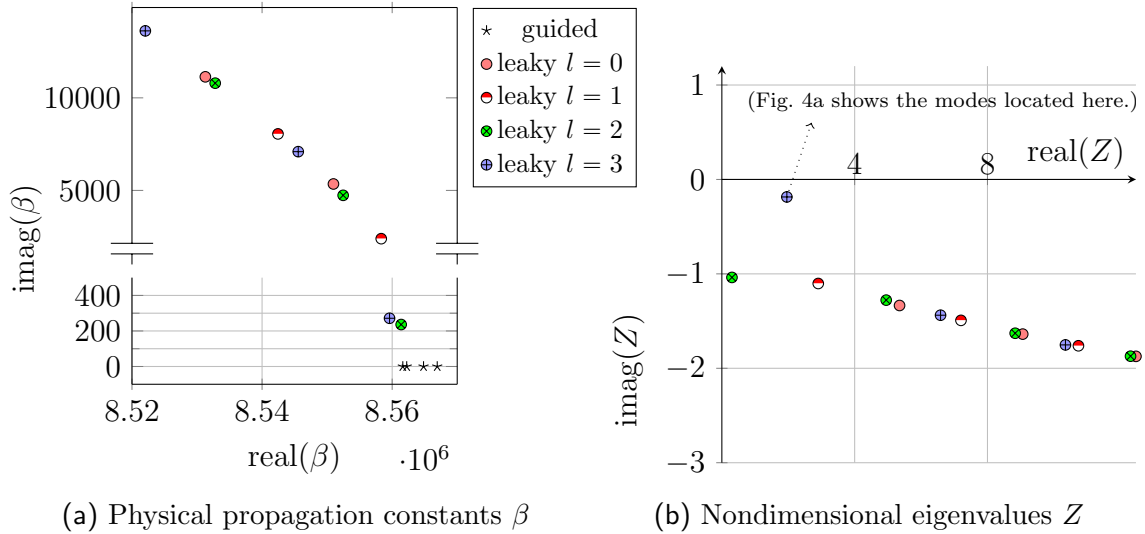

By separation of variables, the solutions
of~\eqref{eq:step-index-nondim-X-core}
and~\eqref{eq:step-index-nondim-Z-clad} take the form
$C_1 \rho_{l} (X\rh) e^{\ii {l}  \theta}$ and
$C_0 \rho_{l} (Z\rh) e^{\ii {l}  \theta},$ respectively, where
$\rho_{l} (s)$ satisfies the Bessel equation
$s^2 d^2\rho_{l} /ds^2 + s d \rho_{l} /ds + (s^2 - {l} ^2)
\rho_{l} =0$. Since the solution must be finite at $\rh=0$ and
outgoing as $\rh \to \infty$, we pick $\rho_{l} $ to be $J_{l} $ and
$H_{l} ^{(1)},$ respectively, in the core and cladding regions. Thus
we obtain the following family of solutions
of~\eqref{eq:step-index-nondim-X-core}--\eqref{eq:step-index-nondim-Z-clad}
indexed by~${l} $:
\begin{equation}
  \label{eq:13}
  \uh(\rh, \theta) = \left\{
    \begin{aligned}
      & C_1 J_{l} (X \rh) e^{\ii {l}  \theta},  && \rh \le 1,\\
      & C_0 H_{l} ^{(1)} (Z \rh) e^{\ii {l}  \theta},  && \rh > 1.
    \end{aligned}
    \right.  
\end{equation}
Note that in the computation of guided modes (see
e.g.,~\cite[Chapter~5]{Reide16}) one chooses the exponentially
decaying Bessel solution $K_{l} $ in the $\rh>1$ region, but to
compute leaky modes, we must instead choose the outgoing Hankel
function, as done above.

The interface conditions of~\eqref{eq:step-index-nondim-interface} can
now be expressed as
\[
  T
	\begin{bmatrix}
		C_1	\\
		C_0	\\
	\end{bmatrix}
	=
	\begin{bmatrix}
		0	\\
		0	\\
              \end{bmatrix},
              \text{ where }
              T = 
	\begin{bmatrix}
		J_{{l} }(X) & -H_{{l} }^{(1)}(Z)	\\
		X J_{{l} }'(X) & -Z(H_{{l} }^{(1)})'(Z)	\\
	\end{bmatrix}.
\]
Nontrivial solutions are obtained when
$\det T = -ZJ_{{l} }(X)(H_{{l} }^{(1)})'(Z) + X
J_{{l} }'(X)H_{{l} }^{(1)}(Z) $ is zero.  Using the 
well-known~\cite{AbramStegun72} Bessel identities
$J_{{l} }'(z) = ({l} /z) J_{{l} }(z) - J_{{l}  + 1}(z)$ and
$(H_{{l} }^{(1)})'(z) = ({l} /z) H_{{l} }^{(1)}(z) - H_{{l}  +
  1}^{(1)}(z),$  the determinant simplifies to
$\det T $ $= Z J_{l} (X) H_{{l} +1}^{(1)}(Z)$ $- X J_{{l} +1}(X)
H_{l} ^{(1)}(Z)$.  Substituting $X = (V_1^2 + Z^2)^{1/2}$, we conclude
that the eigenvalues $Z \in \C$ are zeros of the function
\begin{equation}
  \label{eq:32}
  f(Z)
  =
  Z J_{{l} }\!\left((V_1^2 + Z^2)^{1/2}\right) H_{{l}  + 1}^{(1)}(Z)
  -
  (V_1^2 + Z^2)^{1/2}
  J_{{l}  + 1}\!\left( (V_1^2 + Z^2)^{1/2}\right) H_{{l} }^{(1)}(Z).  
\end{equation}
Once such a $Z$ is found for an integer ${l} $ (there are usually many
for a single ${l} $-value---see Figure~\ref{fig:ew-leaky-guided}), say $Z_{l} $, the
corresponding nondimensional leaky mode is obtained (up to a scalar
factor) by setting $C_1 = H_{{l} }^{(1)}(Z_{l} )$ and
$C_0 = J_{{l} }(X_{l} )$ in~\eqref{eq:13}. Hence the corresponding
physical mode ($u$) and propagation constant ($\beta$) are given using
 $Z_{l} $, by
\begin{equation}
  \label{eq:14}
  \begin{aligned}
  \beta & = (k^2n_0^2 - (Z_{l} /R_0)^2)^{1/2},
  \\
  X_{l}  & = (V_1^2 + Z_{l} ^2)^{1/2},
 \end{aligned}
 \qquad 
 u(r, \theta)
 =
 \begin{cases}
   H_{{l} }^{(1)}(Z_{l} ) J_{{l} }(X_{l}  r/R_0) e^{i{l} \theta}, & r \le R_0,
   \\
   J_{{l} }(X_{l} ) H_{{l} }^{(1)}(Z_{l}  r/R_0) e^{i{l} \theta}, & r > R_0.
 \end{cases}  
\end{equation}
This is the exact solution that will be the basis of our verification.

To proceed, we choose the parameters of a commercially available
ytterbium-doped, step-index fiber (detailed
in~\cite{GopalGrubiOvallParker20}, where its guided modes were
computed by solving a linear selfadjoint eigenproblem).  The fiber has
a core radius of $R_0=12.5 \times 10^{-6}$~m, core index $n_1 = 1.45097$, and
cladding index $n_0=1.44973$.  The typical cladding radius of this
fiber is $16 R_0$ and the typical operating wavelength is 1064~nm, so
we set $k = 2 \pi/1.064 \times 10^6 \: \mathrm{m}^{-1}.$ We computed
several roots of $f$ for various ${l} $ in high precision using
standard root finding methods~\cite{KravaBarel00}.  They are shown in
Figure~\ref{fig:ew-leaky-guided} (where Figure~\ref{subfig:physical}
also shows the locations of the guided mode eigenvalues
from~\cite{GopalGrubiOvallParker20} on the real line for reference).
Next, we apply Algorithm~\ref{alg:polyfeast} to solve the discrete
cubic eigenproblem~\eqref{eq:discreteCubicEVP} and cross-verify the
results obtained with the above-mentioned root-finding approach.

\begin{figure}
  \centering
  \includegraphics[width=0.8\linewidth]{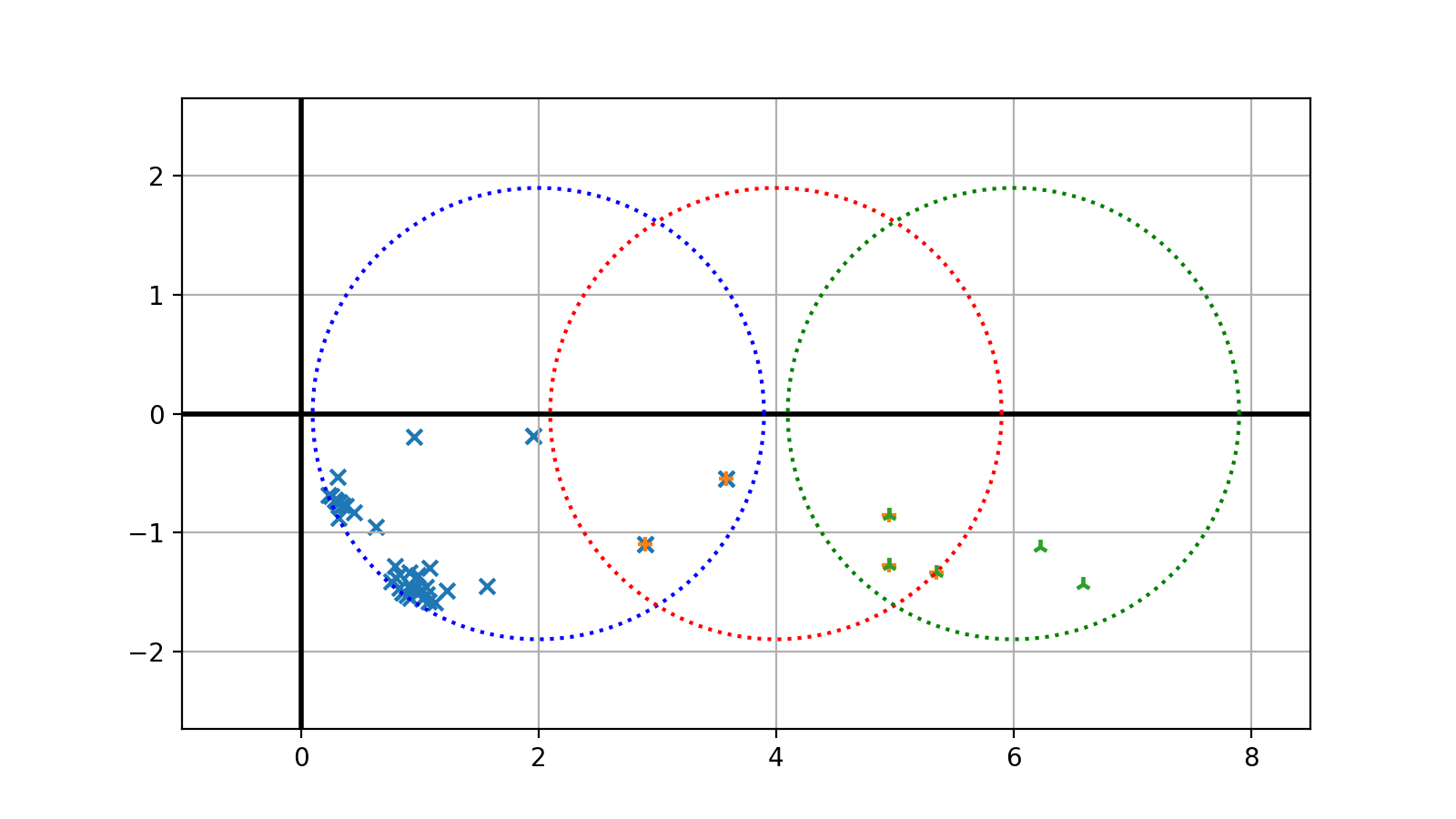}
  \caption{Three collections of Ritz values (marked in different
    styles) produced by ten iterations of
    Algorithm~\ref{alg:polyfeast} using three overlapping contours are
    shown in the complex plane (for $Z$). Numerous values in the first
    contour appear to have arisen from the essential spectrum.}
  \label{fig:search}
\end{figure}

We implemented Algorithm~\ref{alg:polyfeast} as an extension of the
open source finite element library NGSolve \cite{schoberl97,ngsolve}.
The computational parameters used in our numerical studies are
$\Rfin = 4$, $\Rh_0 = 1$ and $\Rh = 2$. One way to apply the algorithm
is to conduct a preliminary ``first search'' using a relatively
large~$m$, large contours, and a fixed number of iterations.  One then
examines the resulting (unconverged) Ritz value locations, identifies the
desired ones by viewing their corresponding subspace iterates, designs
contours that zoom in to a desired eigenvalue location while
separating the remainder, and then runs Algorithm~\ref{alg:polyfeast}
again (with smaller $m$) until convergence.

The results of the above-mentioned first search for this problem are
pictorially illustrated in Figure~\ref{fig:search}. We used three
overlapping circular contours of centers $2, 4,$ and $6$, all of radii
$1.9$ and used the quadrature formula~\eqref{eq:zkwk} with $N=10$.
Setting $\alpha=1$, and $p=5$ we applied ten iterations
of Algorithm~\ref{alg:polyfeast} for each contour, starting with $m=50$
random vectors. The locations of the Ritz values from the algorithm
quickly stabilize, and those finally falling inside the respective
contours are shown in Figure~\ref{fig:search}. Near the southwest edge
of first circular contour, we found numerous Ritz values that appear
to have arisen from the essential spectrum of the undiscretized
operator. Four other values inside that circle are resonances that are
also zeros of the $f$ in~\eqref{eq:32} for some ${l} $. Two of them
coincided with the Ritz values found by the next contour. More values
can be found by adding further contours. Any one of these Ritz values
can now be found by using a tightened contour to find an
eigenvalue to better accuracy, and by running the algorithm again until
convergence. Finer discretizations can be used if needed.
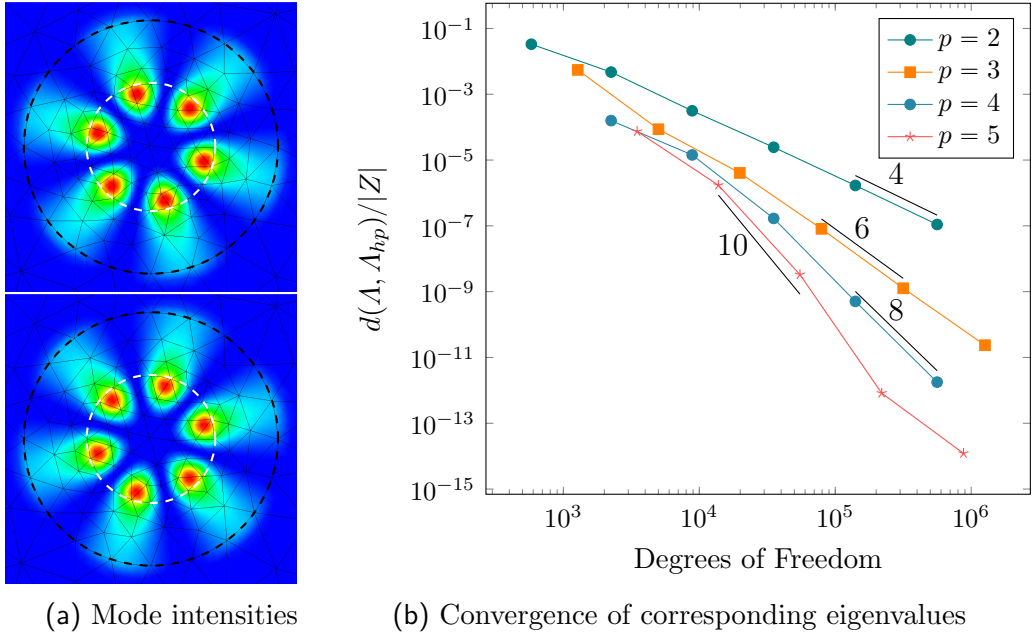
\begin{figure}
  \centering 
	\begin{subfigure}{0.28\textwidth}
          \begin{tikzpicture}
            \node[anchor=south west,inner sep=0] (image) at (0,0) {\includegraphics[width=0.875\textwidth, trim=31.5em 1.5em 31.5em 5.25em, clip]{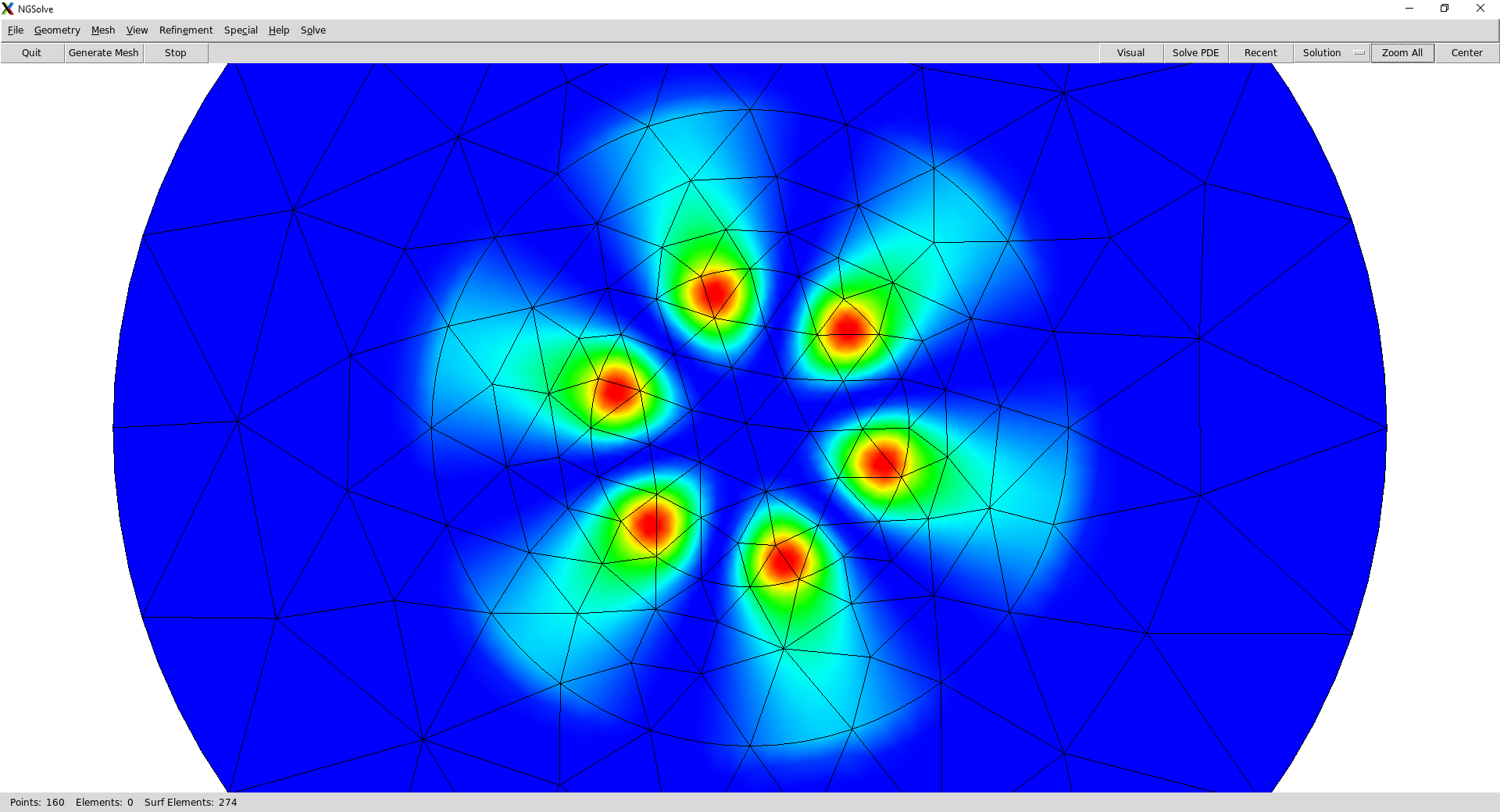}};
            \begin{scope}[x={(image.south east)},y={(image.north west)}]
              \draw[white, opacity=1.0, thick, dashed] (0.5,0.50125) circle (0.220625);
              \draw[black, opacity=1.0, thick, dashed] (0.5,0.50125) circle (0.43625);
            \end{scope}
          \end{tikzpicture}
          \begin{tikzpicture}
            \node[anchor=south west,inner sep=0] (image) at (0,0) {\includegraphics[width=0.875\textwidth, trim=31.5em 1.5em 31.5em 5.25em, clip]{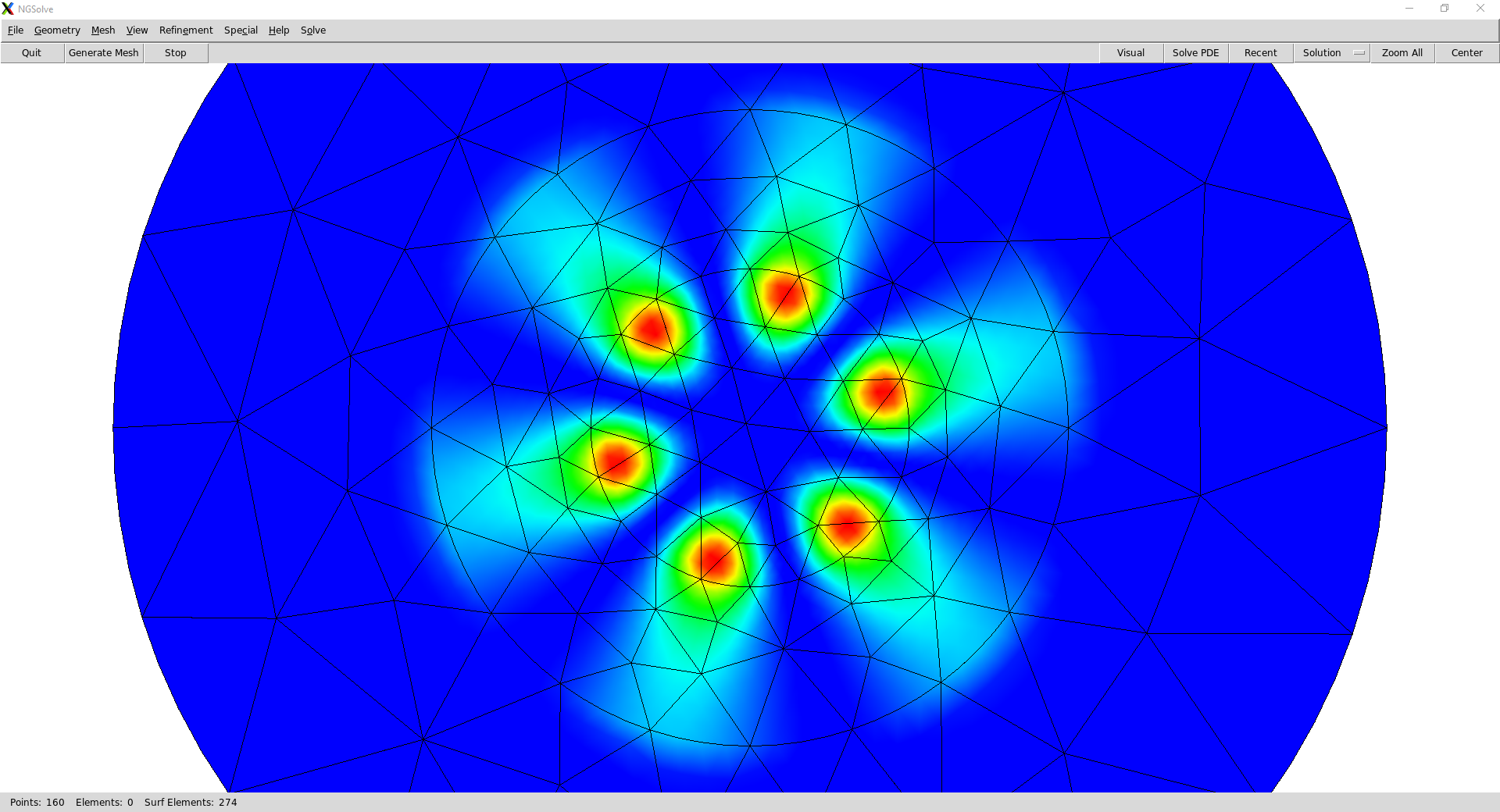}};
            \begin{scope}[x={(image.south east)},y={(image.north west)}]
              \draw[white, opacity=1.0, thick, dashed] (0.5,0.50125) circle (0.220625);
              \draw[black, opacity=1.0, thick, dashed] (0.5,0.50125) circle (0.43625);
            \end{scope}
          \end{tikzpicture}
          \subcaption{Mode intensities}
          \label{subfig:intensities-step-index}
	\end{subfigure}
        %
        %
        \begin{subfigure}{0.56\textwidth}
            \begin{tikzpicture}
              \begin{loglogaxis}[
                footnotesize,
                width=\textwidth,
                height=0.92\textwidth,
                xlabel={Degrees of Freedom},
                ylabel={$d({\Lambda}, {\Lambda_{hp}}) / |Z|$},
                legend pos=north east,
                max space between ticks=30pt,
                cycle list name=exotic
                ]

                \addplot coordinates {
                  (581, 3.32e-02)
                  (2245, 4.71e-03)
                  (8861, 3.14e-04)
                  (35245, 2.45e-05)
                  (140621, 1.69e-06)
                  (561805, 1.10e-07)
                };
                
                \addplot coordinates {
                  (1276, 5.49e-03)
                  (5005, 8.65e-05)
                  (19861, 4.06e-06)
                  (79165, 8.12e-08)
                  (316141, 1.27e-09)
                  (1263565, 2.36e-11)
                };
                
                \addplot coordinates {
                  (2245, 1.58e-04)
                  (8861, 1.43e-05)
                  (35245, 1.68e-07)
                  (140621, 5.06e-10)
                  (561805, 1.78e-12)
                };
                
                \addplot coordinates {
                  (3488, 7.44e-05)
                  (13813, 1.72e-06)
                  (55013, 3.30e-09)
                  (219613, 8.24e-13)
                  (877613, 1.22e-14)
                };
                

                \addplot[color=black] coordinates {
                  (140621, 3.38e-06)
                  (561805, 2.1125e-07)
                } node[midway, above] {$4$};

                \addplot[color=black] coordinates {
                  (79165, 1.624e-07)
                  (316141, 2.5375e-09)
                } node[midway, above] {$6$};

                \addplot[color=black] coordinates {
                  (140621, 1.012e-09)
                  (561805, 3.953125e-12)
                } node[midway, above] {$8$};

                \addplot[color=black] coordinates {
                  (13813, 8.6e-07)
                  (55013, 8.3984375e-10)
                } node[left, midway] {$10$};

                \legend{$p=2$, $p=3$, $p=4$, $p=5$}
              \end{loglogaxis}
            \end{tikzpicture}
            \subcaption{Convergence of corresponding eigenvalues}
            \label{subfig:stepindewf}
          \end{subfigure}
          \caption{
            {\em Left}~(\ref{subfig:intensities-step-index}):
            Intensities of computed step-index leaky modes
            corresponding to two eigenvalues in $\vL_{hp} = \{ 
          Z_{hp}^{(1)}, Z_{hp}^{(2)}\}$ are shown. The white and dark dashed
          curves indicate the core-cladding interface and the start of
          the PML, respectively.
          {\em Right}~(\ref{subfig:stepindewf}): Log-scale plot of the
          distance between exact and approximate eigenvalue
          cluster $\vL_{hp}$
          for polynomial degrees $p =
          2,\ldots, 5$ and  uniform mesh refinements.
        }
\end{figure}

We focus on one of these Ritz values near $c=1.9-0.2\ii$ for further
investigation using a tightened contour.  Setting $\alpha = 8$, $m=5$,
and a circular contour of radius $10^{-1}$ centered around $c$, we run
the algorithm again until convergence. The intensities of the two
modes that Algorithm~\ref{alg:polyfeast} found can be seen in
Figure~\ref{subfig:intensities-step-index}, which plots the square
moduli of two eigenfunctions in the eigenspace resulting from a higher
order computation with $p=10$ on the coarsest mesh. In addition to the
intensity pattern, the curved mesh elements, used to closely
approximate the circular core-cladding interface, as well as the fast
decay of the solution into the PML region are also visible in the same
figure.  The exact eigenvalue $Z$, marked in
Figure~\ref{fig:ew-leaky-guided}, is the one fairly close to the real
axis in the case ${l} =3$, namely $Z \approx 1.957793-0.185432\ii$.
The corresponding physical propagation constant, given by the formula
in~\eqref{eq:14} (also marked in the Figure~\ref{subfig:physical})
is $\beta \approx 8559596.699 + 271.443\ii$. Note that the mode
loss determined by the imaginary part of this $\beta$ is very large,
indicating that this mode is practically useless for guiding energy in
the fiber. Nonetheless, it is a reasonable choice for the limited
purpose of verifying that our numerical method reproduces an
analytically computable leaky mode.

For convergence studies, we repeat the above solution procedure on a
coarse mesh and on its successive refinements. Each refinement is
obtained by connecting the midpoints of the edges of the elements in
the current mesh.  We experiment with polynomial degrees
$p=2,\ldots,5$ in the
discretization~\eqref{eq:discreteCubicEVP}, setting $\alpha=8$. The initial mesh size
$h_0$ corresponds to a coarse mesh with about six elements across the
core (of unit non-dimensional radius) and rapidly increasing element
diameters outside of the core; part of this mesh is visible in
Figure~\ref{subfig:intensities-step-index}. 
In every case, the eigenvalue solver returned a
two-dimensional eigenspace for the contour around $c = 1.9-0.2\ii$.
The computed cluster of Ritz values
$\vL_{hp}$ often contained two distinct numbers near the single
exact $Z$ value, which we enumerate as
$\vL_{hp} = \{ Z_{hp}^{(1)}, Z_{hp}^{(2)}\}$.  Since the exact
eigenvalue cluster  $\vL$ is a
singleton in this case,
the Hausdorff distance between the clusters, denoted by
$d(\vL, \vL_{hp})$, reduces to $\max_{i=1,2} |Z - Z_{hp}^{(i)}|$.  This
distance is normalized by $|Z|$, and the values of $d(\vL, \vL_{hp}) / |Z|$ from
computations using various $h$ and $p$ values are shown graphically in
Figure~\ref{subfig:stepindewf}.

To conclude this verification, observe from
Figure~\ref{subfig:stepindewf} that the eigenvalue errors appear to
approach zero at the rate $O(h^{2p})$. (In the figure, the reference
black lines indicate $O(h^{2p})$ and the values of the corresponding
exponent $2p$  are marked alongside.)  The error values below $10^{-13}$
in Figure~\ref{subfig:stepindewf} are likely not reliable since the
expected errors in the semi-analytically computed value of $Z$ are
also in that neighborhood.

\section{A microstructured fiber}
\label{sec:microfiber}

In this section, we compute the transverse leaky modes of a hollow
core microstructured optical fiber. The microstructure we consider
appear to be known in the folklore to pose severe computational
challenges for PML, although the difficulties are seldom spelled out
in the literature. We proceed to present our computational experience
with Algorithm~\ref{alg:polyfeast} in some detail with the hope that
it may serve as a base for substantive numerical comparisons and
further advancement in numerical methodologies.

The geometry of the fiber is shown in Figure~\ref{fig:microgeom} and
is based on the details given in~\cite{Polet14}.
It consists of six symmetrically placed
thin glass capillaries, of thickness $\tc$, intersecting an outer
glass cladding region, the intersection being characterized by the
embedding distance $\ec$ shown in Figure~\ref{fig:geom-zoom}. The
capillaries and the cladding together form the shaded region in
Figure~\ref{fig:geom}, which we denote by the subdomain $\omsi$.  Let
$\nsi$ and $\nair$ denote the refractive indices of glass and air,
respectively.  The region surrounded by the capillary tubes is the
hollow core region where one would like to guide light.
The mode computation fits within the previously described
model~\eqref{eq:fibermodel}, with $R_0$ as marked in
Figure~\ref{fig:geom}, $n_0 = \nair$, 
and $n_1$ given by the piecewise constant function
\[
  n_1(x) =
  \begin{cases}
    \nsi, & x \in \omsi, \\
    \nair, & x \notin \omsi.
  \end{cases}
\]
In our numerical study, the parameter values  (see
Figure~\ref{fig:geom}) are as follows: $\nair =1.00028$,
$\Rcore = 15 \times 10^{-6}$m,
$R_0=60.775 \times 10^{-6}$m,
$\tclad=10^{-5}$m, 
$\tc = 0.028 \Rcore$, $\ec = 0.001\bar{6} \Rcore$, $\dc = 5 \tc$, $\Rci=0.832 \Rcore$,
and $\Rco=0.86 \Rcore$.
The wavenumber is set by $k = 2\pi/\lambda$, where
the wavelength in vacuum denoted by $\lambda$,
can take different values, two of
which considered below are 1000 and 1800 nanometers. At these two
wavelengths, the values of $\nsi$ are 1.44982 and 1.43882, respectively.
Setting the
characteristic length scale by $L=\Rcore$, we nondimensionalize the
eigenproblem, as described previously, to the
form~\eqref{eq:non-dim-form} and terminate the geometry at
$\Rout=7.385$ nondimensional units.

\begin{figure}
  \centering
  \includegraphics[width=0.9\linewidth]{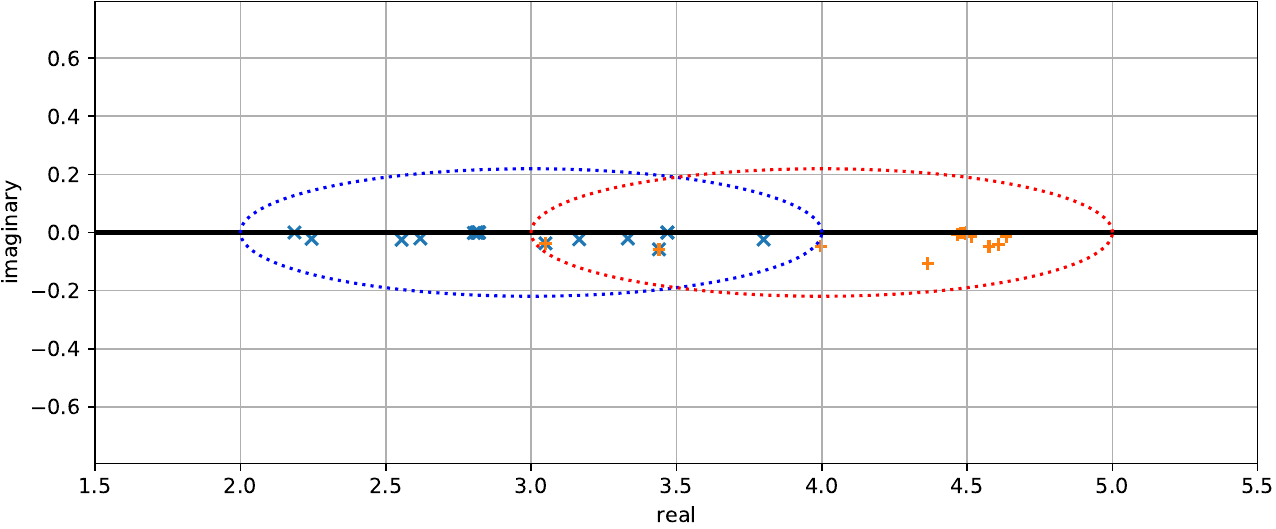}
  \caption{Ritz values found within two overlapping elliptical
    contours in the $Z$-plane for the microstructured fiber.
  }
  \label{fig:search-ellipse}
\end{figure}
\begin{figure}
  \centering
  \begin{subfigure}{0.3\linewidth}
  \includegraphics[trim=350  220    350   290, clip, width=\textwidth]
  {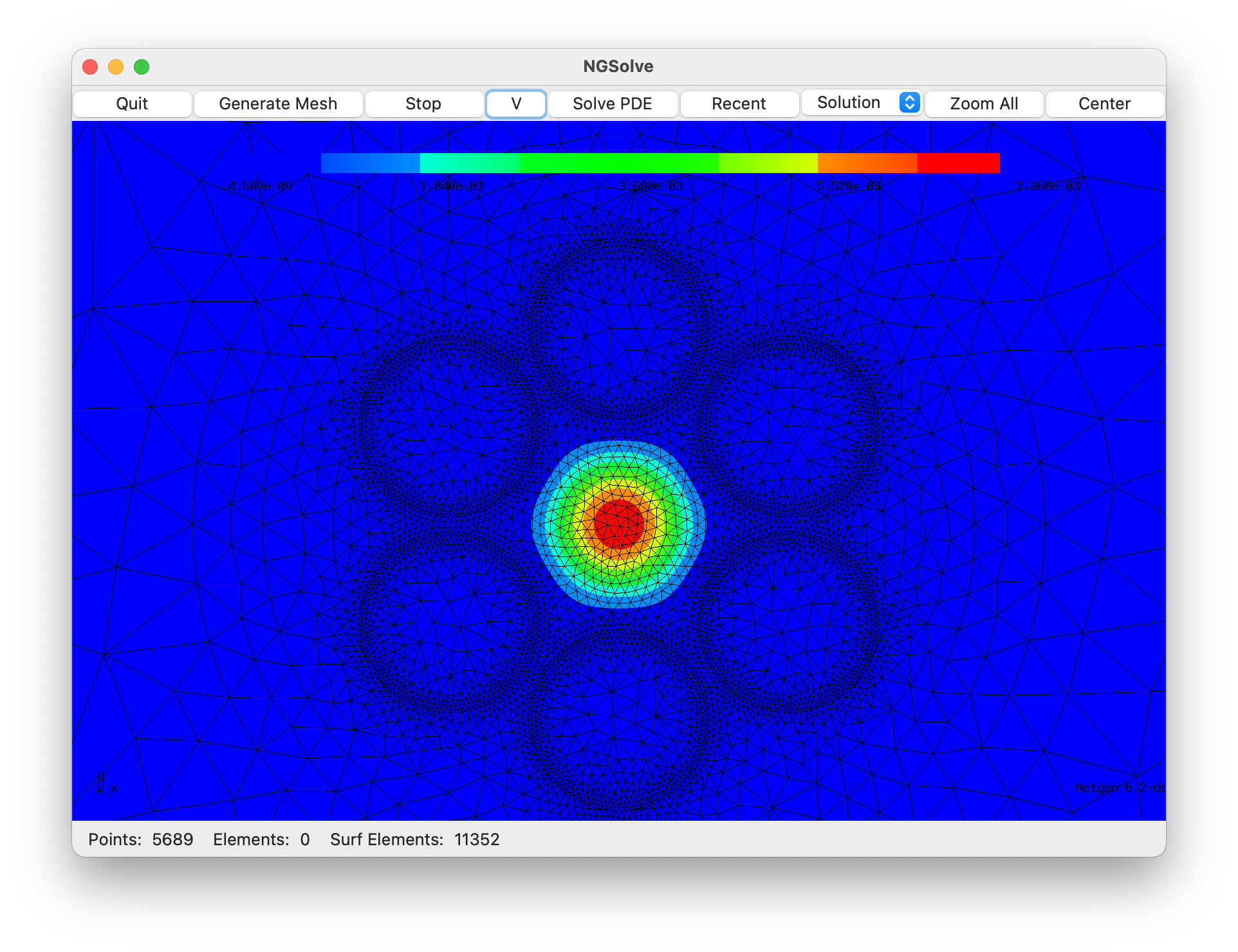}
    \subcaption{$2.186-\ii\, 2.1\times 10^{-6}$}
    \label{subfig:LP01}
  \end{subfigure}
  \begin{subfigure}{0.3\linewidth}
  \includegraphics[trim=350  220    350   290, clip, width=\textwidth]
  {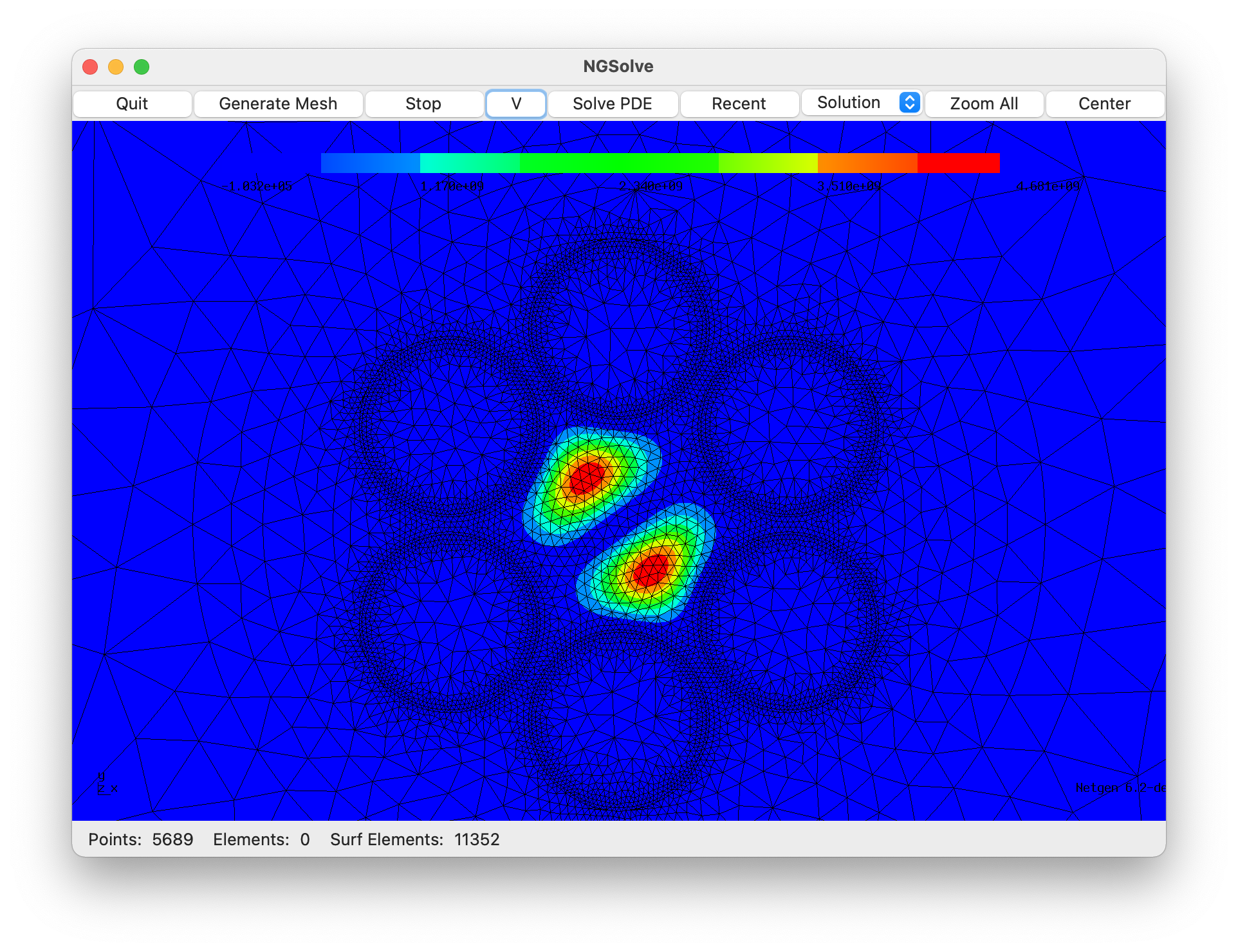}
  \subcaption{$3.469-\ii \,7.0\times 10^{-5}$}
    \label{subfig:LP11a}    
  \end{subfigure}
  \begin{subfigure}{0.3\linewidth}
  \includegraphics[trim=350  220    350   290, clip, width=\textwidth]
  {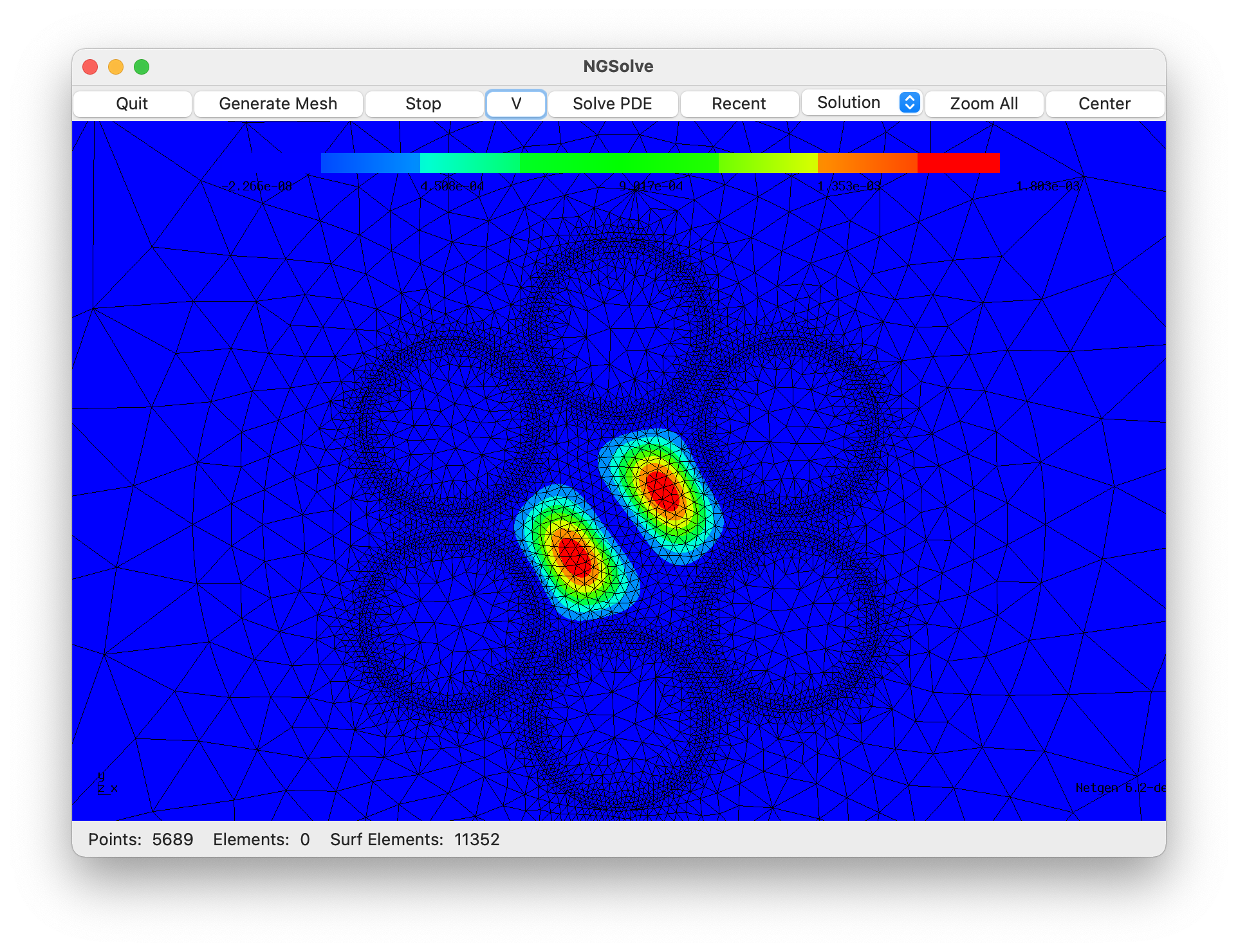}
    \subcaption{$3.469-\ii 6.4\times 10^{-5}$}
    \label{subfig:LP11b}    
  \end{subfigure}  
  \begin{subfigure}{0.3\linewidth}
  \includegraphics[trim=350  220    350   290, clip, width=\textwidth]
  {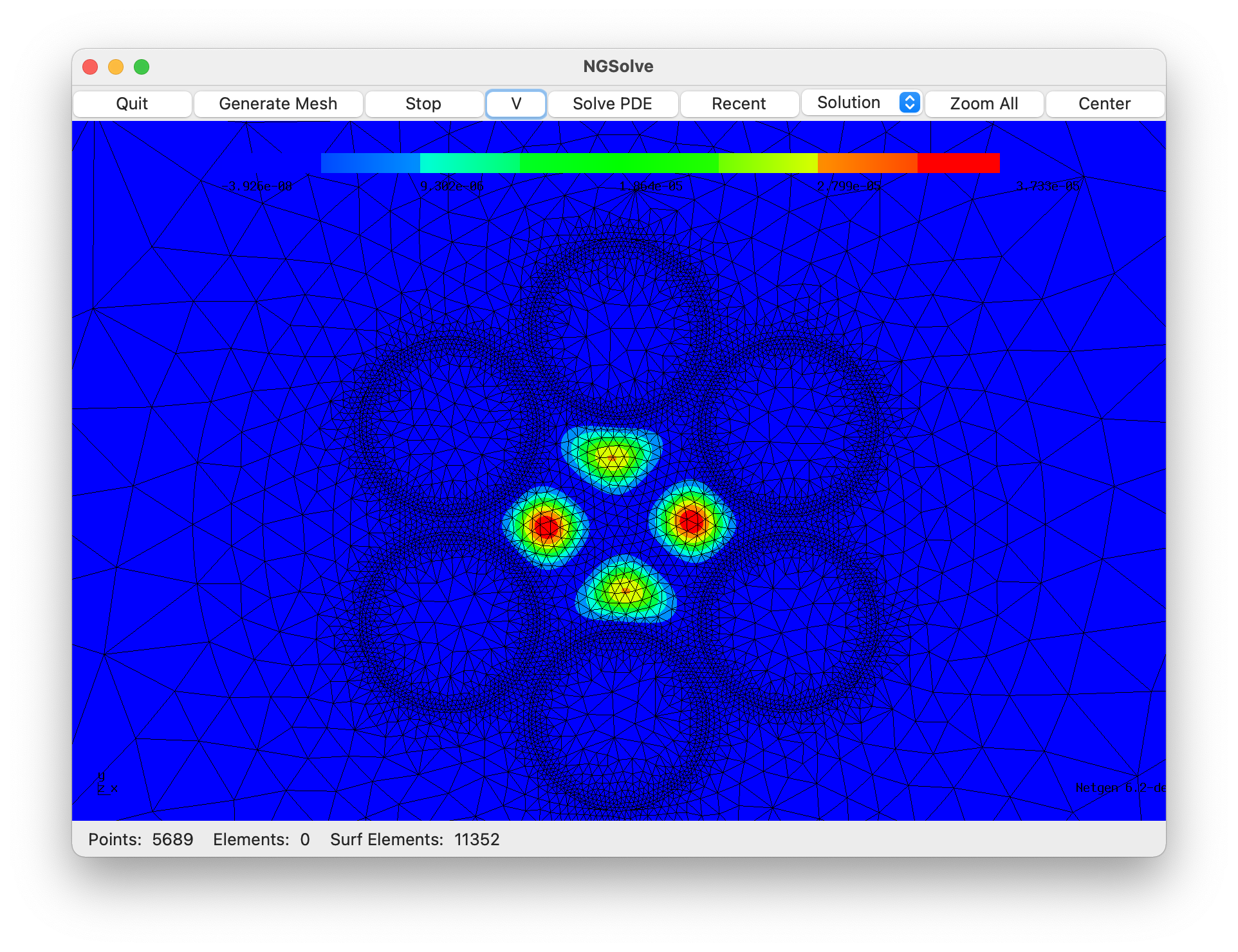}
    \subcaption{$4.637-\ii\,2.4\times 10^{-3}$}
    \label{subfig:LP21a}
  \end{subfigure}       
  \begin{subfigure}{0.3\linewidth}
  \includegraphics[trim=350  220    350   290, clip, width=\textwidth]
  {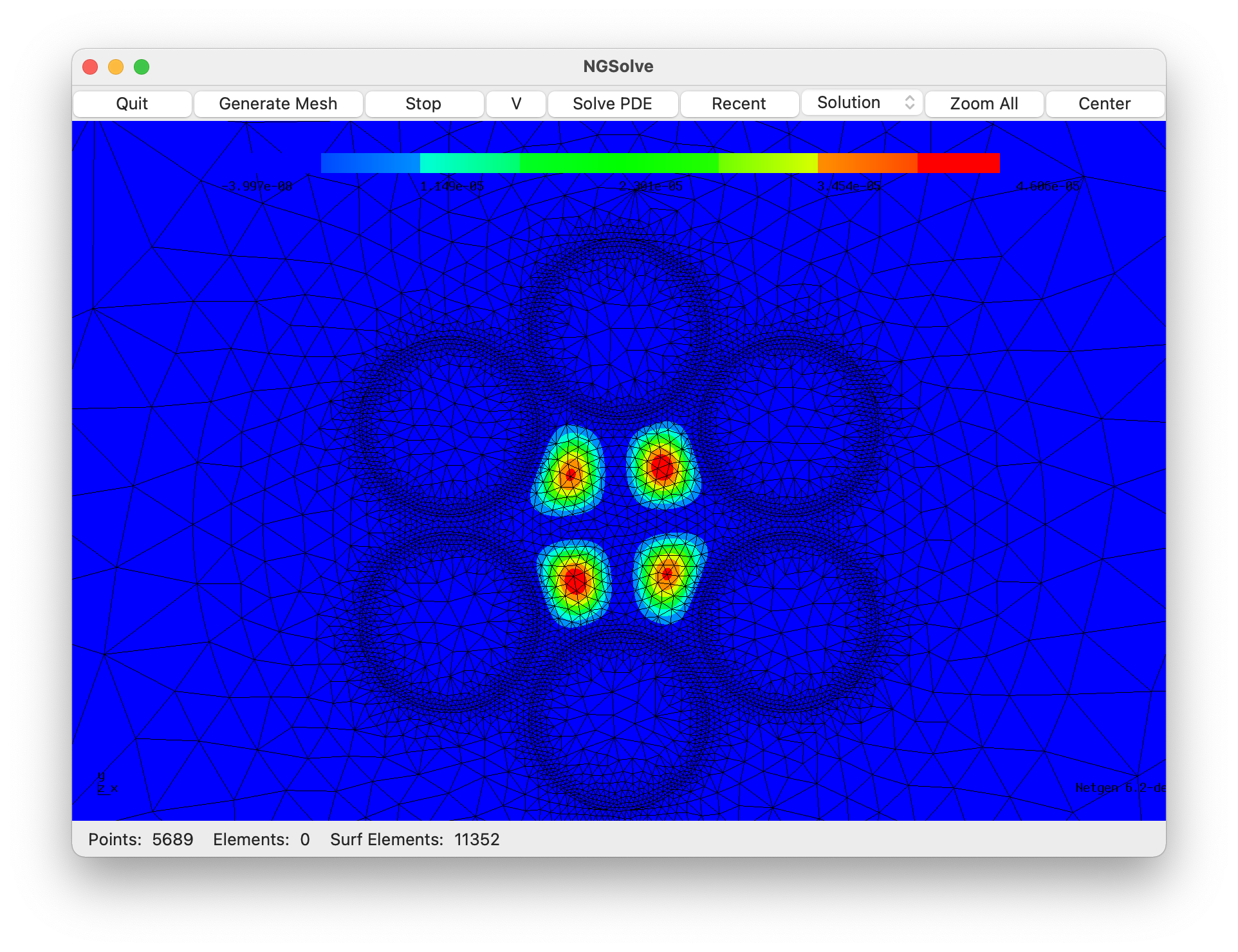}
    \subcaption{$4.637-\ii\,2.2\times 10^{-3}$}
    \label{subfig:LP21b}    
  \end{subfigure}
  \begin{subfigure}{0.3\linewidth}
  \includegraphics[trim=350  220    350   290, clip, width=\textwidth]
  {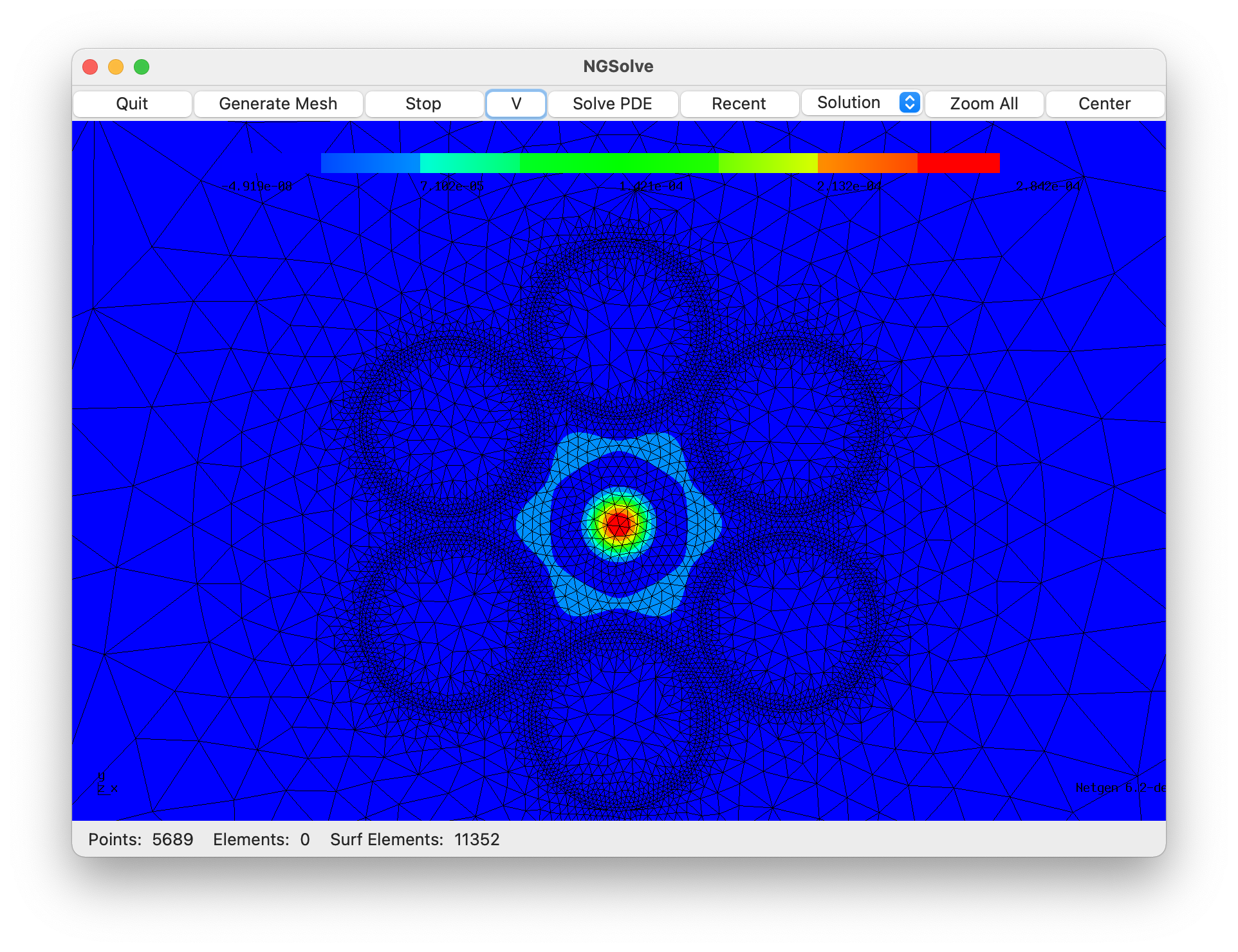}
    \subcaption{$4.961-\ii\,8.6\times 10^{-4}$}  
    \label{subfig:LP21b}    
  \end{subfigure}
  \caption{Intensities of computed modes are  shown zoomed into a
    rectangle covering the hollow core (the region $r<\Rcore$ of
    Figure~\ref{fig:microgeom}), labeled with their approximate
    nondimensional $Z$ values for $\lambda=10^{-6}$m.}
  \label{fig:modes_poletti}
\end{figure}

\begin{figure}
  \centering
  \begin{subfigure}{0.3\linewidth}
    \includegraphics[trim=285  100    285   180, clip, width=\textwidth]
     {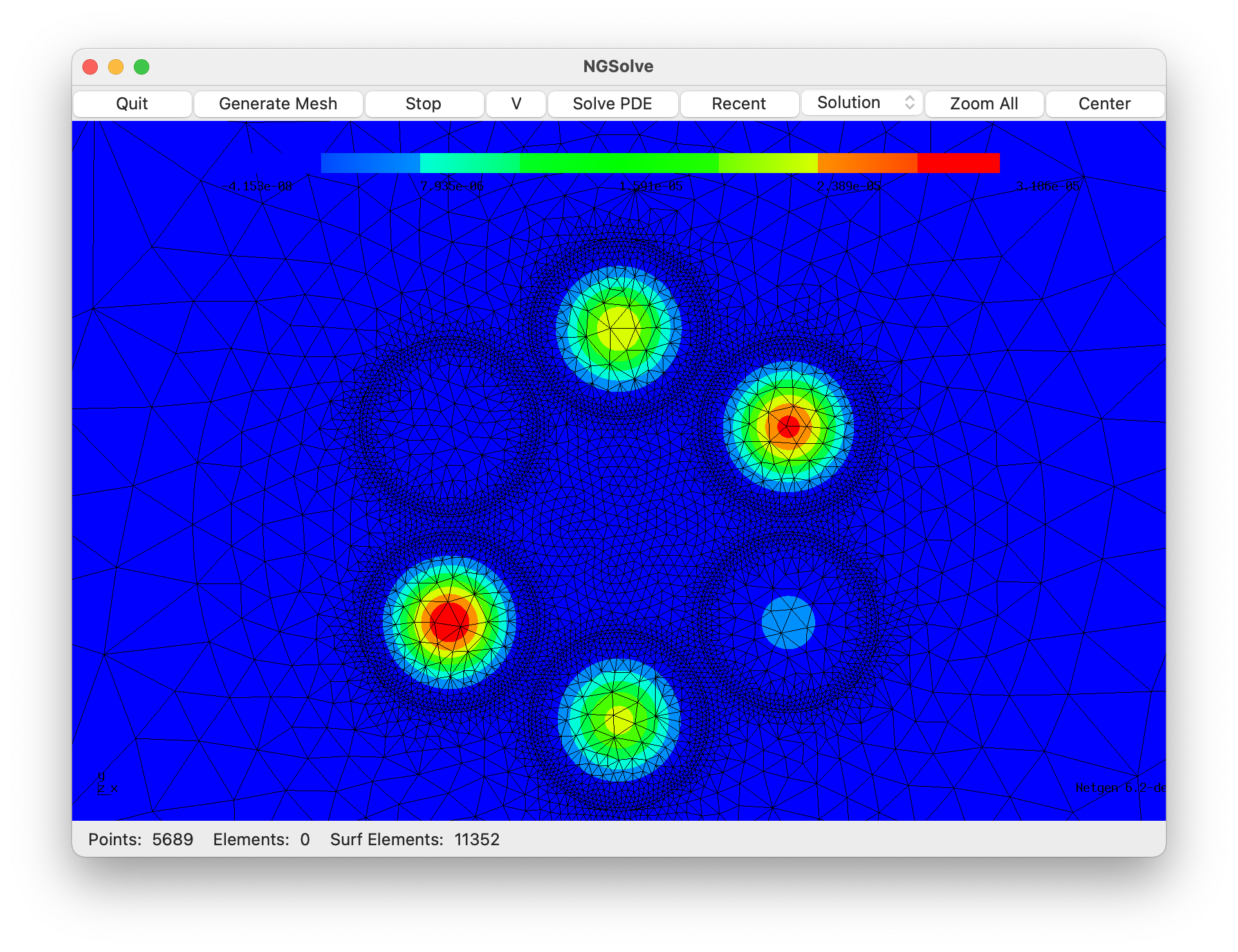}
     \subcaption{$2.802-\ii\,2.4\times 10^{-3}$}
     \label{subfig:cap00}    
  \end{subfigure}
  \begin{subfigure}{0.3\linewidth}
    \includegraphics[trim=285  100    285   180, clip, width=\textwidth]
    {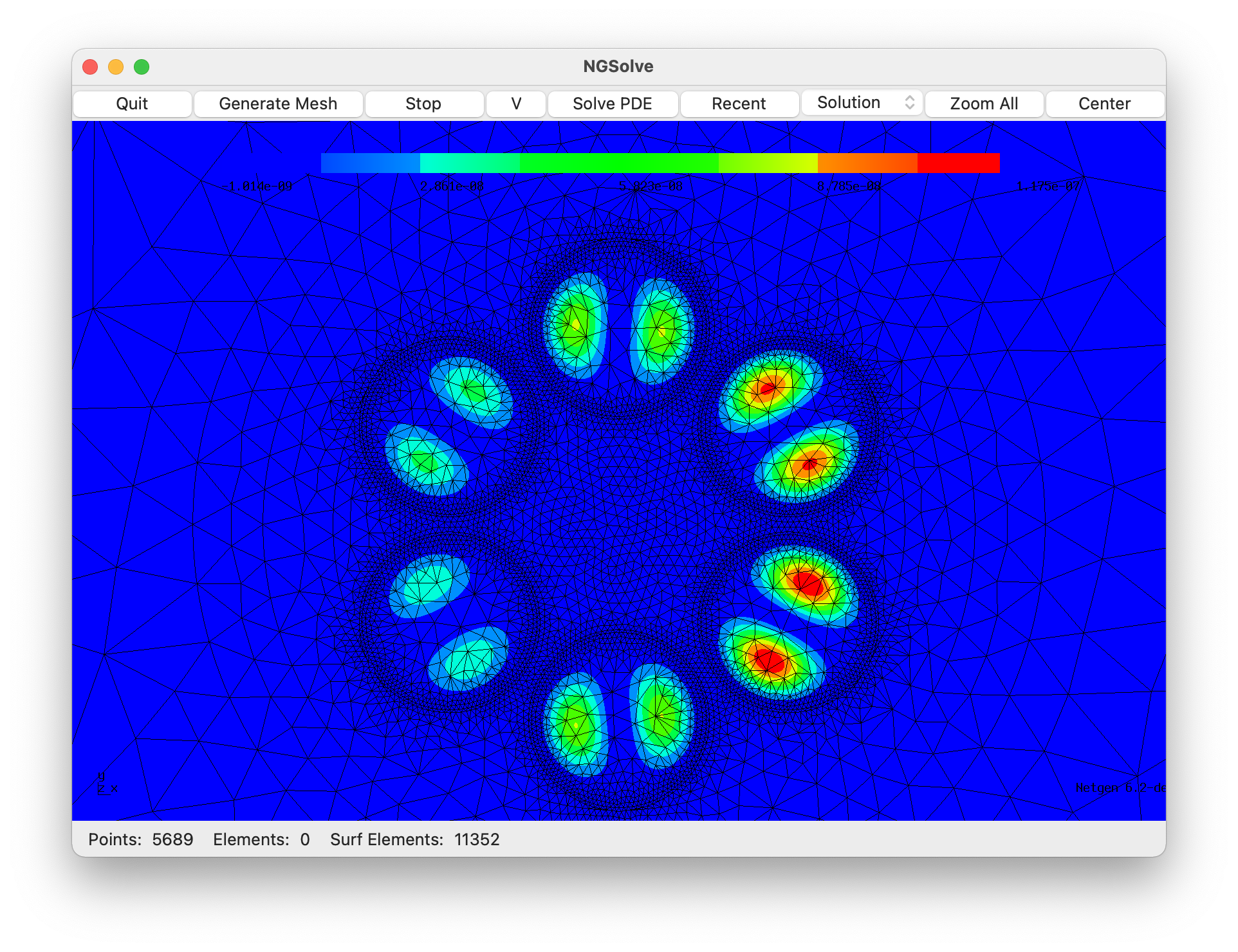}
    \subcaption{$4.481-\ii\,7.4\times 10^{-3}$}
     \label{subfig:cap01}    
  \end{subfigure}
  \caption{Examples of high-loss non-core modes observed in
    computations are shown. Six modes similar to that in
    Fig.~\ref{subfig:cap00} and twelve modes similar to that in
    Fig.~\ref{subfig:cap01} were observed.}
  \label{fig:highlossmodes}
\end{figure}

Next, we apply the method verified in Subsection~\ref{ssec:verif}. As
described there, we first apply Algorithm~\ref{alg:polyfeast} to
conduct a preliminary search, followed by further runs to obtain
accurately converged eigenvalues.  Before we give the results of the
preliminary search, recall from Figure~\ref{fig:search} that
discretizations of the essential spectrum can seep into circular
contours close to the origin, wasting computational resources on
irrelevant modes.  In the current example, we show how to avoid this
using elliptical contours. Since the eigenvalues arising from the
essential spectrum are expected to subtend a negative acute angle with
the real axis at the origin, an elliptical contour (by increasing its
eccentricity) can avoid them better than circular contours. This is
seen in our results of Figure~\ref{fig:search-ellipse} (where we do
not see the signs of essential spectrum that we saw in
Figure~\ref{fig:search}). The Ritz values in this figure were output
after a few iterations of Algorithm~\ref{alg:polyfeast} employing the
quadrature formula in~\eqref{eq:zkwk-ellipse} with two overlapping
ellipses centered at $y = 3$ and $y=4$, each with $\gamma=1$,
$\rho^{-1}=0.8$, $N = 10$, and $m=20$. For this figure, the discrete
nonlinear eigenproblem was built using $\lambda=1000$~nm, $\alpha=5$
(the default value of $\alpha$ used for all computations in this section), 
$p=10$, and a mesh with curved elements sufficient to resolve the thin
geometrical features. A part of this mesh is visible in
Figure~\ref{fig:modes_poletti}.  We will refine this mesh many times
over for some computations below.

Let us now focus on one of these Ritz values near $Z=2.18$ and run
Algorithm~\ref{alg:polyfeast} to convergence using a tight circular
contour that excludes all other Ritz values. The corresponding
computed leaky mode is shown in Figure~\ref{subfig:LP01} for the case
$p=20$.  Although the algorithm quickly converges for various
discretization parameters to (visually) the same eigenmode, we
observed a surprisingly {\em large preasymptotic} regime where the
imaginary part of the eigenvalues varied significantly even as meshes
were made finer and polynomial degrees were increased. Convergence was
observed only after crossing this preasymptotic regime. We proceed to
describe its implication on estimating mode loss, an important
practical quantity of interest.  Confinement loss (CL) in fibers,
usually expressed in decibels (dB) per meter, refer to
$-10\log_{10}(\wp(1)/\wp(0))$ where $\wp(z)$ is the power at the length
$z$~meters. For a leaky mode, viewing
$\wp(z)$ as proportional to $|e^{-\ii\beta z}|^2$,
its CL can be estimated (see e.g., \cite[pp.~213]{Reide16})
from the propagation constant by
CL$=-20 \log_{10} e^{-\Im \beta} = 20 \Im\beta / \ln(10)$.

\begin{figure}
  \centering
  \includegraphics[width=\textwidth]{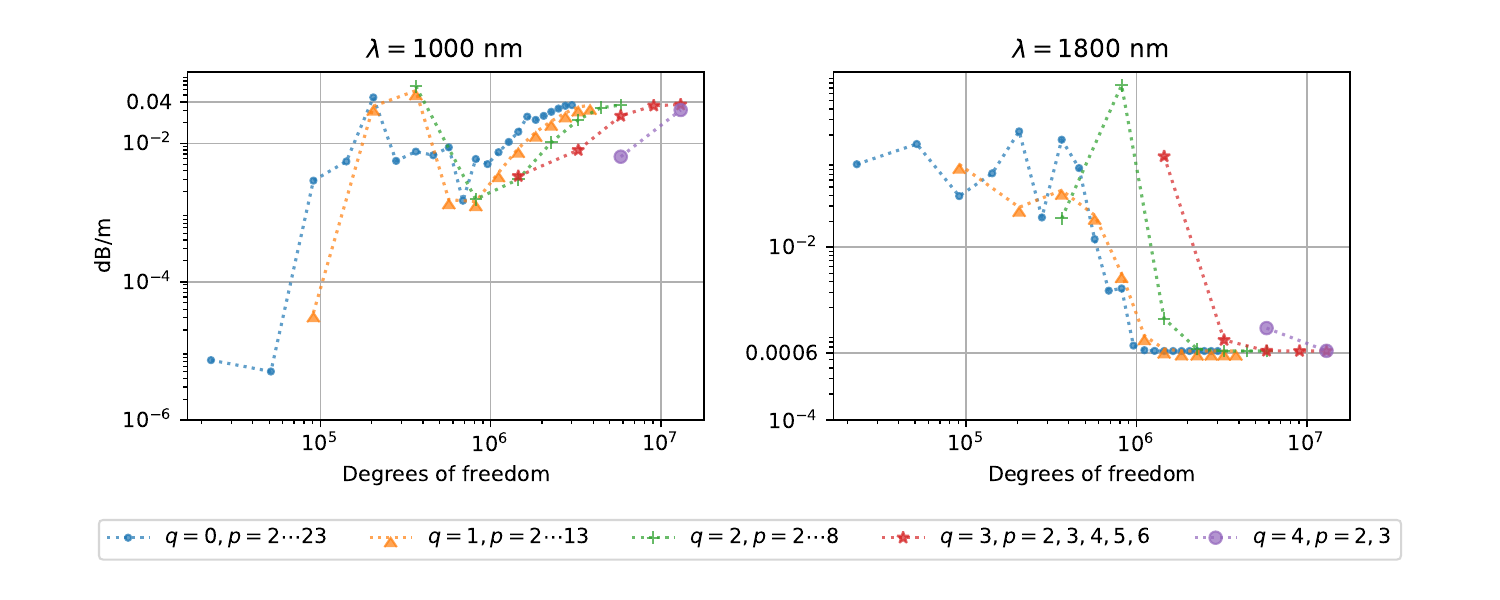}
  \caption{Computed confinement losses for the hollow core fiber
    show prominent preasymptotic variations for lower mesh refinements
    ($q$) and polynomial degrees $(p)$.}
  \label{fig:CL}
\end{figure}

We computed CL from the eigenvalues obtained for various $h$ and
$p$. To systematically vary $h$, we started with an initial mesh (part
of which is visible in
Figures~\ref{fig:modes_poletti} and \ref{fig:highlossmodes}) and
performed successive refinements. One refinement divides each
triangular element in the mesh into four (and the four are exactly
congruent when the element is not curved). Thus, the mesh after ${q}$
refinements has a grid size $2^{-{q}}$ times smaller than the
starting mesh.  The results for various ${q}$ and degrees $p$ are in
Figure~\ref{fig:CL}. Each dotted curve there represents many
computations performed on a fixed mesh (i.e., fixed ${q}$-value) for 
increasing values of the degree~$p$. For the
case of $\lambda=1000$~nm we observe that the computed CL values
appear to converge to around 0.04 dB/m, but only well
after a few millions of degrees of freedom. In particular, the CL
values computed using discretizations with under one million degrees
of freedom are {\em off by a few orders} of magnitude.  We also
observe that quicker routes (more efficient in terms of degrees of
freedom) to converged CL values are offered by the choices that use
higher degrees~$p$ (rather than higher mesh refinements ${q}$).  The
second plot in Figure~\ref{fig:CL} shows similar results for the case
of operating wavelength $\lambda=1800$~nm. In this case, CL seems to be
largely overestimated in a preasymptotic regime, but computations
using upwards of several millions of degrees of freedom agree on a
value of around CL$=0.0006$~dB/m, as seen from Figure~\ref{fig:CL}.

We have also confirmed that
our results remain stable in the asymptotic
convergent regime as we vary the PML parameters.
Table~\ref{tab:clstudy} shows
an example of results from one such parameter
variation study.
Remaining in the above-mentioned case of 
$\lambda = 1800$~nm, we focus on how one of the points (for $q=1,
p=10$)
in the second
plot of Figure~\ref{fig:CL} varies
under changes in PML width and strength.
Table~\ref{tab:clstudy} displays how 
the computed CL values vary slightly around $0.000628$~dB/m.

\begin{table}
	\centering
	\begin{tabular}{|c|c|c|c|c|}
          \hline
          & & \multicolumn{3}{c|}{CL (dB/m)}
          \\
          \cline{3-5}
          PML width ($\mu$m) &	Degrees of freedom
            &$\alpha = 1$ & $\alpha = 5$ & $\alpha = 10$ \\ \hline
		 50 & 2270641 & 0.000630 & 0.000629 & 0.000629 \\ 
		100 & 2603121 & 0.000628 & 0.000628 & 0.000629 \\ 
		150 & 2482041 & 0.000628 & 0.000628 & 0.000628 \\ \hline
	\end{tabular}
	\caption{Computed confinement losses for varying PML widths
          and PML parameters $\alpha$ (fixing $q=1$ and  $p = 10$).}
	\label{tab:clstudy}
\end{table}

We conclude this section
by presenting visualizations of modes through plots of
their intensities (which are proportional to the square moduli) of the
computed modes.  In Figure~\ref{fig:modes_poletti}, in
addition to the fundamental mode (Figure~\ref{subfig:LP01}), we
display a few further higher order modes we found
(Figures~\ref{subfig:LP11a}---\ref{subfig:LP21b}). Although these higher
order modes exhibit good core localization, their CL values are higher
than the fundamental mode. This hollow core structure also admits
modes that support transmission outside of the central hollow core,
such as those shown in Figure~\ref{fig:highlossmodes}. They are, however,
much more lossy than the fundamental mode.

\section{Proofs}
\label{sec:proof}

A basic ingredient for proving Theorem~\ref{thm:projlin} is the
next result which can be found
proved using differential equations
in~\cite[Chapter~7]{GohLanRod82}. We give a different elementary
proof. Let $N$ denote a $k\times k$ nilpotent
matrix with zero entries except for ones on the first
superdiagonal, and let $I$ be the $k\times k$ identity matrix.
Then $J=\lambda I + N $ is a $k \times k$ Jordan matrix.

\begin{lemma}
  \label{lem:chain}
  A sequence $v_0, v_1, \ldots v_{k-1}$ in $\C^n$ is a
  nontrivial Jordan chain of a nonlinear eigenvalue $\lambda$
  of $P(z)$ in
  the sense of definition~\eqref{eq:18}, if and only if $v_0 \ne 0$
  and $V = [v_0, v_1, \ldots v_{k-1}] \in \C^{n \times k}$ satisfies
  \begin{equation}
    \label{eq:16a}
    \sum_{i=0}^d A_i V J^i = 0.
  \end{equation}
\end{lemma}
\begin{proof}
  Let $s_{i\ell} = (\begin{smallmatrix} i \\ \ell
      \end{smallmatrix}) \,\lambda^{i-\ell}A_i V N^\ell$.  The sum
      in~\eqref{eq:16a} can be alternately expressed as
  \begin{align*}
    \sum_{i=0}^d A_i V J^i
    & = \sum_{i=0}^d A_i V
    \sum_{\ell=0}^{\min(i, k-1)}
    \begin{pmatrix}
      i \\ \ell
    \end{pmatrix}
    \lambda^{i-\ell} N^\ell
    =\left( \sum_{i=0}^{k-1} 
    + \sum_{i=k}^{d}
    \right)  \sum_{\ell=0}^{\min(i, k-1)}
    s_{i\ell}
    \\
    & =
      \left(\sum_{i=0}^{k-1} \sum_{\ell=0}^i
      + \sum_{i=k}^d \sum_{\ell=0}^{k-1}\right)
      s_{i\ell}
    = 
      \left(\sum_{\ell=0}^{k-1} \sum_{i=\ell}^{k-1}
      + \sum_{\ell=0}^{k-1} \sum_{i=k}^d\right) s_{i\ell}
      = \sum_{\ell=0}^{k-1} \sum_{i=\ell}^d s_{i\ell}.
  \end{align*}
  Observing the connection between the summands and the 
  derivatives of $P(\lambda)$, 
  \[
    \sum_{i=\ell}^d
    s_{i\ell} =
    \frac{1}{\ell!} \big[
    \,0\cdots 0_\ell, \;
    P^{(\ell)}(\lambda) v_0, \ldots,
    P^{(\ell)}(\lambda) v_{k-1-\ell}\big].
  \]
  where $0\cdots 0_\ell$ denote $\ell$ zero columns.  
  Hence~\eqref{eq:16a} is equivalent to
  \[
    \sum_{\ell=0}^{k-1}\frac{1}{\ell!} \big[
    \,0\cdots 0_\ell, \;
    P^{(\ell)}(\lambda) v_0, \ldots,
    P^{(\ell)}(\lambda) v_{k-1-\ell}\big] = 0.
  \]
  Since the $j$th column of the left hand side is the same as
  $\sum_{\ell=0}^j (1/\ell!) P^{(\ell)}(\lambda) v_{j-\ell},$ these
  equations are exactly the same as the defining requirements for
  $v_i$ to form a Jordan chain for a nonlinear eigenvalue $\lambda$,
  per definition~\eqref{eq:18}.
\end{proof}

\begin{proof}[Proof of Theorem~\ref{thm:projlin}]
  Let $\lambda$ be a nonlinear eigenvalue enclosed by $\vG$. Then, by
  Lemma~\ref{lem:chain}, $v_0, v_1, \ldots v_{k-1}$ in $\C^n$ is an
  associated (right) Jordan chain if and only if equation~\eqref{eq:16a}
  holds, which is the same as the {\em last} equation of the following
  system
    \begin{equation}
    \label{eq:3}    
  \begin{bmatrix}
    0     &  I  &         & \\
    \vdots&     & \ddots  & \\
    0     &     &         & I \\
    A_0   & A_1 & \cdots  & A_{d-1}
  \end{bmatrix}
  \underbrace{
    \begin{bmatrix}
      V \\ VJ \\ \vdots \\ VJ^{d-1}
    \end{bmatrix}
  }_{\Vc}
    =
    \begin{bmatrix}   
      I    & \\
      & \ddots\\
      &        & I & \\
      &        &   & -A_d
    \end{bmatrix}
    \underbrace{
      \begin{bmatrix}
        V \\ VJ \\ \vdots \\ VJ^{d-1}
      \end{bmatrix}}_{\Vc} J.
  \end{equation}
  Considering that the remaining equations of~\eqref{eq:3} trivially
  hold, we have shown that \eqref{eq:16a} holds if and only
  if~\eqref{eq:3} holds.  Let $\Vc_i \in \C^{n d}$ denote the
  $(i+1)^{\text{th}}$ column of the $nd \times k$ matrix $\Vc$
  indicated in~\eqref{eq:3}, so that
  $\Vc = [ \Vc_0, \;\Vc_1, \ldots, \Vc_{k-1}]$.  Since~\eqref{eq:3} is
  the same as $\Ac \Vc = \Bc \Vc J$, or equivalently (see
  the characterization~\eqref{eq:17})
  \[
    (\Ac - \lambda \Bc) \Vc_0 = 0, \;\text{ and } (\Ac - \lambda \Bc)
    \Vc_i = \Bc \Vc_{i-1} \;\text{ for } i=1, 2,\ldots, k-1,
  \]  
  i.e., the columns of $\Vc$ form a Jordan
  chain for the linear matrix pencil $\Ac - \lambda \Bc$.  
  Noting that $v_i = F \Vc_i$, we have thus shown that
  $\Vc_0, \ldots, \Vc_{k-1} \in \C^{nd}$ is a Jordan chain for
  $\Ac - \lambda \Bc$ if and only if its {\em first} blocks, namely
  $v_0, v_1, \ldots v_{k-1}$ in $\C^n,$ form a Jordan chain for
  $P(\lambda)$. Since any $x \in \C^n$ falling within the algebraic
  eigenspace of the nonlinear eigenvalue $\lambda$ is a linear
  combination of such chains $\{v_i\}$, the proof of the first item
  of the theorem is complete.

  To prove the second item of the theorem, we start as above with the
  nonlinear eigenvalue $\lambda,$ but now proceed with its {\em left}
  Jordan chain $\vt_0, \vt_1, \ldots \vt_{k-1} \in \C^n.$ The second
  definition of~\eqref{eq:18} implies
  $\sum_{l=0}^j (l!)^{-1} \Pt^{(l)}(\bar\lambda) \vt_{j-l} = 0$ where
  $\Pt(z) = \sum_{j=0}^d z^j A_j^*$. Applying Lemma~\ref{lem:chain} to
  $\Pt(z)$, we find that
  $\Vt = [\vt_0, \vt_1, \ldots \vt_{k-1}] \in \C^{n \times k}$
  satisfies
  \begin{equation}
    \label{eq:34}
    \sum_{i=0}^d A_i^* \Vt \bar J^i =0.
  \end{equation}
  Simple calculations show that~\eqref{eq:34} holds if and only if
  \begin{equation}
    \label{eq:35}
    \begin{bmatrix}
      0  &\cdots & 0     & A_0^*  \\
      I  &\ddots &\vdots & \vdots \\
         &\ddots & 0     & A_{d-2}^*\\
         &       & I     & A_{d-1}^*\\
    \end{bmatrix}
    \begin{bmatrix}
    W_0 \\ \vdots \\ W_{d-2} \\ \Vt  
  \end{bmatrix}
  =
  \begin{bmatrix}   
    I    & \\
    & \ddots\\
    &        & I & \\
    &        &   & -A_d^*
  \end{bmatrix}
  \begin{bmatrix}
    W_0 \\ \vdots \\  W_{d-2} \\ \Vt  
  \end{bmatrix}
  \bar J
  \end{equation}
  with $W_{d-i} = -\sum_{j=0}^{i-1} A^*_{d-j} \Vt \bar{J}^{i-1-j}$ for
  $i=2, \ldots d$. Taking the conjugate transpose of both sides
  of~\eqref{eq:35}, we find that
  $\Vct^* = [W_0^*, \cdots, W_{d-2}^*, \Vt^*] \in \C^{k \times nd}$
  satisfies $\Vct^* \Ac = J' \Vct^* \Bc,$ an identity which when
  written using the columns $\Vct_i \in \C^{nd}$ of 
  $\Vct=[\Vct_0,\ldots, \Vct_{k-1}]$ reads
  \[
    \Vct_0^*(\Ac - \lambda \Bc) =0,
    \text{ and }
    \Vct_i^*( \Ac- \lambda \Bc) = \Vct_{i-1}^* \Bc.
  \]
  Keeping~\eqref{eq:17adj} in view, the above equivalences have thus shown
  that the columns of $\Vct$ form a left Jordan chain of
  $\Ac - \lambda \Bc$ if and only if the the columns of the {\em last}
  block of $\Vct$, namely $L\Vct = \Vt = [\vt_0, \ldots, \vt_{k-1}]$
  form the left Jordan chain of $\lambda$ as a nonlinear eigenvalue of
  $P(z)$.
\end{proof}

\begin{proof}[Proof of Theorem~\ref{thm:resolvent}]
  The given $X$  satisfies $(z\Bc - \Ac)X = Y.$ In block component form,
  this yields
  \begin{subequations}
    \begin{gather}
      \label{eq:10}
      z X_{i-1} - X_i   = Y_{i-1}, \qquad i=1, 2, \ldots, d-1,
      \\ \label{eq:11a}
      -A_0 X_0 - A_1 X_1 - \dots - A_{d-2} X_{d-2} - ( z A_d + A_{d-1})
      X_{d-1} = Y_{d-1}.
    \end{gather}    
  \end{subequations}
  Clearly~\eqref{eq:10} is the same as~\eqref{eq:Xcomps-i}. Moreover, 
  the $i^{\text{th}}$ equation of~\eqref{eq:10}, when combined with the
  $(i-1)^{\text{th}}$ equation  of~\eqref{eq:10}, yields
  $
  X_i = z X_{i-1} - Y_{i-1} = z(z X_{i-2} - Y_{i-2}) - Y_{i-1}.
  $
  This process can be recursively continued to get
  \begin{gather*}
    X_i = z^i X_0 -
    \sum_{j=0}^{i-1} z^{i-1-j} Y_j,
    \qquad i=1, 2, \ldots, d-1.
  \end{gather*}
  Substituting these expressions for $X_i$ for $i \ge 1$
  into~\eqref{eq:11a}, we obtain
  \[
    \sum_{i=0} ^{d-1} A_i
    \bigg(
    z^i  X_0 -
    \sum_{j=0}^{i-1} z^{i-1-j} Y_j,
    \bigg)
    = -Y_{d-1}
    -z A_d
    \bigg(z^{d-1} X_0 -
  \sum_{j=0}^{d-2} z^{d-2-j} Y_j\bigg).
\]
Moving the term with $z^d$ from right to left, we identify a group of
terms that sum to $P(z) X_0$. Sending
all the remaining terms on the left to the right, 
simplifying, and applying $P(z)^{-1}$ to both sides, we
obtain~\eqref{eq:Xcomps-0}.

For proving~\eqref{eq:Xcomps-adj}, let us rewrite the equation
$(z\Bc - \Ac)^* X= W$ as the system of equations
\begin{subequations}
  \label{eq:978}
  \begin{align}
    \label{eq:9}
    \bar z {\Xt}_0 - A_0^* {\Xt}_{d-1} & = W_0
    \\ \label{eq:7}
    -{\Xt}_{i-1} + \bar z {\Xt}_i - A_{i}^* {\Xt}_{d-1}
                                       & = W_i, \qquad i=1, 2, \ldots, d-2,
    \\ \label{eq:8}
    -{\Xt}_{d-2} - (\bar z A_d^* + A_{d-1}^*) {\Xt}_{d-1}
                                       & = W_{d-1},
  \end{align}
\end{subequations}
Multiply~\eqref{eq:7} by $\bar z^i$
and add up all the equations of~\eqref{eq:9}--\eqref{eq:7}. Then
observe that all terms of the type $-\bar z^{i} \Xt_{i-1}$ telescopically
cancel off in the resulting sum, yielding
\[
  \bar z^{d-1} {\Xt}_{d-2}
  = W_0 + \bar z W_1 + \bar z^2 W_2 + \cdots + \bar z^{d-2} W_{d-2}
  + (A_0^* + \bar z A_1^* 
  + \cdots + \bar z^{d-2} A_{d-2}^*)
  {\Xt}_{d-1}.
\]
Using~\eqref{eq:8} to eliminate ${\Xt}_{d-2}$ from the last equation and
rearranging, we have
\[
  \sum_{j=0}^{d-1} \bar z^j W_j + P(z)^* {\Xt}_{d-1} = 0,
\]
which
immediately yields the expression for ${\Xt}_{d-1}$
in~\eqref{eq:Xcomps:d-1}. The expressions for the 
remaining ${\Xt}_i$ in~\eqref{eq:Xcomps-adj}
follow from~\eqref{eq:8} and~\eqref{eq:7}, respectively.
\end{proof}

\section{Conclusion}
\label{sec:conclusion}

We have presented a new technique to compute collections
of transverse leaky modes of
optical fibers with complex microstructure by combining advances in
contour integral eigensolvers and frequency-dependent PML.  The
frequency-dependent PML is not yet widely used for resonance
computations due to the difficulties in solving the resulting
nonlinear eigenproblem. The new avenue we presented using
Algorithm~\ref{alg:polyfeast} makes it a viable computational option
when a few resonance values (enclosed in a given contour) is of
interest.

Algorithm~\ref{alg:polyfeast} is 
applicable more generally to solve for clusters of
nonlinear eigenvalues of the polynomial type arising from any
application.  The efficiencies in the algorithm were gained by
circumventing the typical large inverses arising from linearization of
the polynomial eigenproblem.

We have exploited FEAST's flexibility with contours to design ellipses
that effectively probe wanted resonances without interference from the
deformed essential spectrum. The algorithm also eliminates unwanted
eigenfunctions supported in the PML region by sending them to the
eigenspace of infinity.

While confinement loss values for some fiber geometries (such as that
in Subsection~\ref{ssec:verif}) can be computed easily and fast, the
antiresonant fiber we considered in Section~\ref{sec:microfiber}
presented a preasymptotic regime, which was surprisingly large for a
two-dimensional structure, where computed loss values vary by orders
of magnitude. By reporting this in detail, we hope to bring more
awareness of this issue to those estimating losses  of similar
structures with thin filaments.

\bigskip

\subsection*{Acknowledgements} 

We gratefully acknowledge extensive discussions with Dr.~Jacob Grosek
(Directed Energy, Air Force Research Laboratory, Kirtland, NM) on
practical issues with the accurate computation of transverse modes of
optical fibers
and with Dr.~Markus Wess (ENSTA, Paris) on his dissertation research.
This work was supported in part by
AFOSR grant FA9550-19-1-0237, AFRL Cooperative Agreement 18RDCOR018,
and NSF grant DMS-1912779.



\end{document}